\documentclass[aps,prl,10pt,twocolumn,superscriptaddress,floatfix,nofootinbib,showpacs,longbibliography]{revtex4-2}
\usepackage{makecell}
\usepackage{array}
\newcolumntype{C}[1]{>{\centering\arraybackslash}m{#1}}
\usepackage{booktabs,adjustbox}
\usepackage{float}
\usepackage{amsmath}
\usepackage[utf8]{inputenc}  
\usepackage[T1]{fontenc}     

\usepackage[table]{xcolor}
\usepackage{color,soul}
\usepackage[british]{babel}  
\usepackage[normalem]{ulem}

\usepackage[colorlinks=true, citecolor=magenta, urlcolor=blue]{hyperref}  
\usepackage{graphicx} 
\usepackage[babel]{microtype}  
\usepackage{amsmath,amssymb,amsthm,bm,amsfonts,mathrsfs,bbm} 

\usepackage{xspace}  
\usepackage{xcolor}
\usepackage{multirow}
\usepackage{array}
\usepackage{bigstrut}
\usepackage{braket}
\usepackage{color}
\usepackage{natbib}
\usepackage{multirow}
\usepackage{mathtools}
\usepackage{float}
\usepackage{blkarray}
\usepackage[caption = false]{subfig}
\usepackage{xcolor,colortbl}
\usepackage{color}
\usepackage{rotating}

\usepackage{tabularx}
\usepackage{tikz}
\usetikzlibrary{calc}

\newcommand{\Tr}{\operatorname{Tr}}

\newcommand{\be}{\begin{equation}}
\newcommand{\ee}{\end{equation}}
\newcommand{\ba}{\begin{eqnarray}}
\newcommand{\ea}{\end{eqnarray}}

\newcommand{\tr}{\operatorname{Tr}}

\newtheorem{theorem}{Theorem}
\newtheorem{corollary}{Corollary}
\newtheorem{definition}{Definition}
\newtheorem{proposition}{Proposition}
\newtheorem{observation}{Observation}

\newtheorem{construction}{Construction}
\newtheorem*{example*}{Example}
\newtheorem{remark}{Remark}
\newtheorem{lemma}{Lemma}

\newcommand{\Span}{{\mathsf{Span}}}

\def\>{\rangle}
\def\<{\langle}

\newtheorem{fact}{Fact}

\usepackage{fontawesome5}
\usepackage{eso-pic}
\usepackage{charter}

\usepackage{centernot}
\usepackage{subfig}
\usepackage{enumitem}

\newcommand{\ko}[1]{\textcolor{blue}{\uline{#1}}}

\newcommand{\SNBose}{Department of Physics of Complex Systems, S. N. Bose National Centre for Basic Sciences, Block JD, Sector III, Salt Lake, Kolkata 700 106, India.}

\newcommand{\ISI}{Physics and Applied Mathematics Unit,  Indian Statistical Institute Kolkata, 203 B.T. Road, 700108, India.}

\newcommand{\IITJ}{Indian Institute of Technology Jodhpur, Jodhpur 342030, India.}

\begin{document}

\title{Emergence and Recovery of ({\it logical}) Kochen–Specker Contextuality via Hamilton Extension}

\author{Jayashree Karmakar}\affiliation{\SNBose}
\author{Biswadeep Chatterjee}\affiliation{\SNBose}
\author{Rafiuddin Gazi}\affiliation{\SNBose}
\author{Anushko Chattopadhyay}\affiliation{\IITJ}
\author{Ananya Chakraborty}\affiliation{\SNBose}
\author{Snehasish Roy Chowdhury}\affiliation{\ISI}
\author{Amit Mukherjee}\affiliation{\IITJ}
\author{Manik Banik}\affiliation{\SNBose}

\begin{abstract}
\noindent Logical Kochen–Specker (KS) contextuality is widely regarded as an intrinsic property of specially constructed measurement configurations. We show instead that it can emerge from KS-colorable vector sets through a constructive procedure we call the Hamilton extension. Defined for four-dimensional vector sets, the Hamilton extension associates each real vector with a measurement context while inducing additional measurement contexts among Hamilton-extended children of distinct parent vectors. These emergent contexts fundamentally alter the compatibility structure, transforming KS-colorable configurations into KS-uncolorable ones and recovering logical contextuality lost under apex-vertex augmentation. We establish a sharp and optimal threshold: the Hamilton extension of every five-vector parent set remains KS-colorable, whereas suitably chosen six-vector parent sets already generate logical KS contradictions. Thus, six vectors constitute the smallest parent set capable of generating KS contradiction through this mechanism. Our results reveal a new structural route to contextuality, provide a systematic framework for constructing compact KS sets, and have implications for contextuality-based quantum information protocols and graph-theoretic approaches to nonclassicality.
\end{abstract}

\maketitle

\medskip
\noindent The Kochen–Specker (KS) theorem establishes that a quantum system of dimension three or higher defies any deterministic noncontextual hidden-variable explanation \cite{Kochen1967}—a single-system counterpart to Bell's no-go theorem of local realism for spatially separated systems \cite{Bell1964,Bell1966}. These no-go results and their experimental confirmations \cite{Freedman1972,Aspect1982-1,Aspect1982-2,Zukowski1993,Bartosik2009,Kirchmair2009,Amselem2009,Lapkiewicz2011} mark some of the most striking departures of quantum theory from classical worldviews \cite{Mermin1993,Brunner2014,Budroni2022}. KS contextuality manifests in a hierarchy of increasing strength, ranging from state-dependent contextuality (SD-C) through state-independent contextuality (SI-C) to the strongest form—logical KS contextuality—where no predetermined outcome assignment satisfies the \emph{exclusivity} and \emph{completeness} constraints on quantum measurements \cite{Budroni2022}. Beyond its foundational significance, contextuality has emerged as a key resource for diverse quantum-information tasks, including quantum computation, communication complexity, cryptography, randomness certification, self-testing, and zero-error communication \cite{Bechmann2000,Cubitt2010,Abbott2012,Howard2014,Vega2017,Bravyi2018,Bharti2019,Alimuddin2023,Gupta2023,Agarwal2026,Ambuj2026,Chattopadhyay2026}. These operational applications, together with its foundational significance, motivate a deeper investigation into the structural mechanisms governing the emergence, disappearance, and recovery of contextuality.

\noindent The KS proofs begin with collections of quantum measurements whose orthogonality relations already encode a contextual contradiction. In this paradigm, contextuality is revealed rather than generated. A fundamental and largely unexplored question is whether contextuality itself can be engineered: can one start from an entirely KS-colorable configuration and, by systematically extending its measurement compatibility structure, generate logical contextuality? Closely related is the problem of recovering contextuality once it has been eliminated by graph-theoretic operations. For example, if $\mathtt{G}$ is a logical KS graph, adjoining an apex vertex connected to every vertex of $\mathtt{G}$ produces the graph $\mathtt{G}+1$, which is always KS colorable and therefore eliminates the original logical contradiction. These questions address whether logical contextuality is an intrinsic property of a measurement configuration or can instead emerge through systematic modifications of its compatibility structure.

\noindent Here, we answer both questions affirmatively. Inspired by Hamilton's quaternion algebra~\cite{Hamilton1844}, we introduce a constructive framework, termed the \emph{Hamilton extension}, that associates a measurement context with every real vector in a four-dimensional parent set. Crucially, the extension induces additional, previously absent, measurement contexts through orthogonality relations between extensions of different parent vectors. These emergent contexts fundamentally alter the compatibility structure of the parent configuration, transforming previously KS-colorable vector sets into logically Kochen--Specker contextual ones. Our main results are threefold. First, we establish that logical KS contextuality can be generated—not merely revealed—by systematically extending measurement compatibility relations. In particular, we transform several KS-colorable parent sets into logically contextual configurations, thereby obtaining new KS-uncolorable vector sets. Second, we show that the same mechanism provides a systematic route for recovering logical KS contextuality destroyed by apex-vertex augmentation. Third, we establish a sharp threshold theorem: the Hamilton extension of every five-vector parent set remains KS-colorable, whereas a suitably chosen six-vector parent set already generates logical KS contextuality. Thus, six is the minimum parent-set cardinality for which this emergence mechanism can operate. Collectively, these results identify Hamilton extension as a structural mechanism for the emergence and recovery of logical KS contextuality. Beyond advancing our understanding of contextuality, this framework provides a systematic route to constructing compact KS sets and offers new directions for contextuality-based quantum information protocols and graph-theoretic characterizations of nonclassicality.

\begin{figure}[t!]
\centering
\includegraphics[width=1\linewidth]{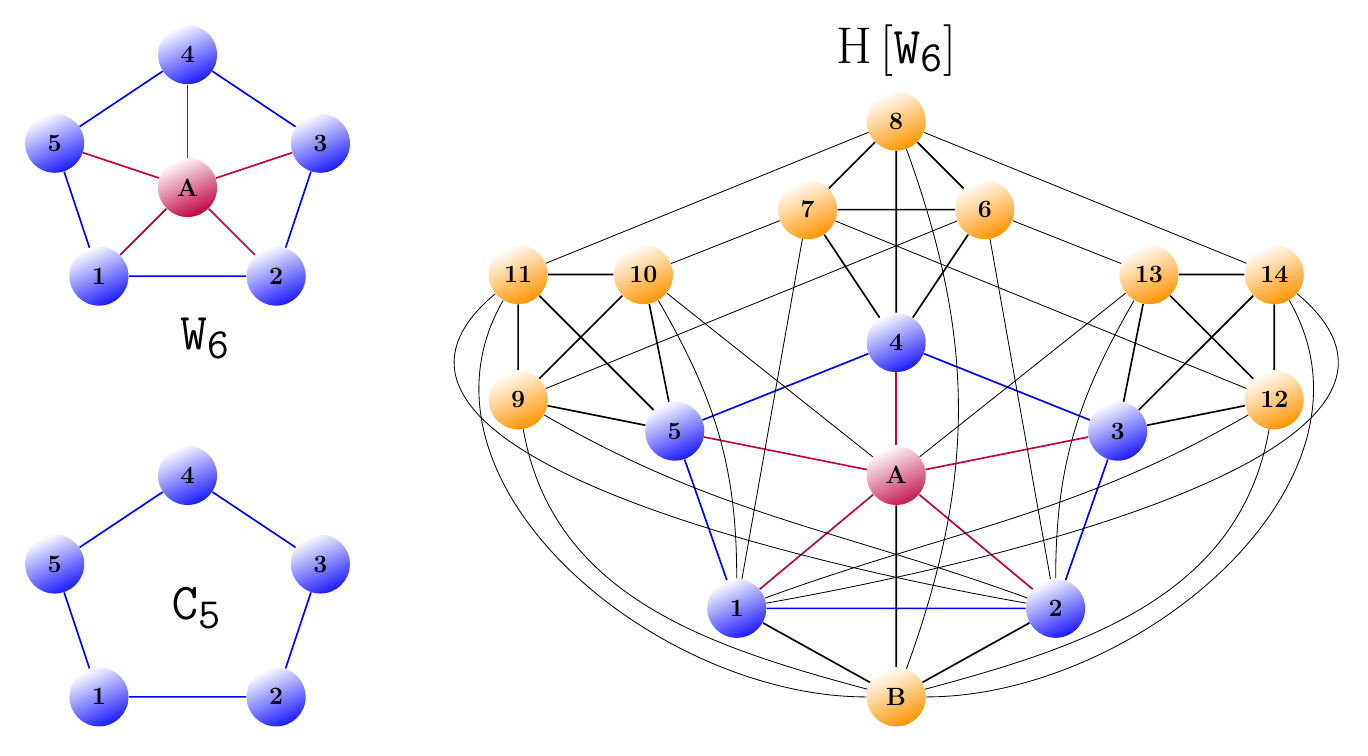}
\hrule\vspace{-.35cm}
\begin{align*}
&\hspace{.5cm}\mathtt{C}_{5}
\scriptstyle \equiv\big\{\tilde{v}_1=(0,1,0);~\tilde{v}_2=(0,0,1);~\tilde{v}_3=(1,-1,0);~\tilde{v}_4=(1,1,1);~\tilde{v}_5=(1,0,-1)\big\};\\ 
&~~~\mathtt{W}_{6}\scriptstyle\equiv
\left\{
\begin{aligned}
&\scriptstyle\Vec{v}_1=\mathtt{E}(\tilde{v}_1)=(0,1,0,0);~ \Vec{v}_2=\mathtt{E}(\tilde{v}_2)=(0,0,1,0);~\Vec{v}_3=\mathtt{E}(\tilde{v}_3)=(1,-1,0,0);\\
&\scriptstyle\Vec{v}_4=\mathtt{E}(\tilde{v}_4)=(1,1,1,0);~~\Vec{v}_5=\mathtt{E}(\tilde{v}_5)=(1,0,-1,0);~~ \Vec{v}_A=(0,0,0,1).
\end{aligned}
\right\};\\
&\mathrm{H}\left[\mathtt{W}_{6}\right]\scriptstyle \equiv\left\{
\begin{aligned}
& \scriptstyle\mathrm{H}[\Vec{v}_{\star}]=\left\{\begin{aligned}&\scriptstyle\text{Computational basis}\\
&\scriptstyle\text{vectors for}~ \star=1,2,A
\end{aligned}\right\},\scriptstyle\mathrm{H}[\Vec{v}_4]=\left\{
\begin{aligned}
&\scriptstyle(1,1,1,0),~~ (-1,1,0,1),\\
&\scriptstyle(1,0,-1,1),(0,-1,1,1)
\end{aligned}
\right\}\\
&\scriptstyle\mathrm{H}[\Vec{v}_3]=\big\{(1,\pm1,0,0),(0,0,1,\pm1)\big\},\mathrm{H}[\Vec{v}_5]=\big\{(1,0,\pm1,0),(0,1,0,\pm1)\big\}
\end{aligned}
\right\}.
\end{align*}
\vspace{-.25cm}\hrule\vspace{-.25cm}
\caption{(Color online) Illustration of apex augmentation and Hamilton extension. The cycle graph $\mathtt{C_{5}}$ (bottom left) admits an orthogonal representation in $\mathbb{R}^{3}$. Adding an apex vertex adjacent to all vertices of $\mathtt{C_{5}}$ yields the wheel graph $\mathtt{W_{6}}$ (top left), which admits an orthogonal representation in $\mathbb{R}^{4}$. The Hamilton extension $\mathrm{H}[\mathtt{W_{6}}]$ is shown on the right. Each parent vertex generates a Hamilton context, while additional edges give rise to emergent contexts, e.g. $\left\{\vec v_B,\vec v_2,\vec v_9,\vec v_{11}\right\}$.}\vspace{-.5cm}
\label{fig1}
\end{figure}

\begin{definition}\label{def1}
(Kochen--Specker set) A KS set is a finite collection of vectors $\mathcal{K}\subset\mathbb{C}^d$, equivalently rank-one projectors $\mathrm{P}_{\mathcal{K}}:=\{{\bf v}\equiv[\vec v]: \vec v\in\mathcal{K}\}$, in $d\ge3$, for which there exists no assignment $\mu:\mathrm{P}_{\mathcal{K}}\to\{0,1\}$ satisfying: (i) \textit{Exclusivity:} $\mu({\bf v})+\mu({\bf w})\le 1$ whenever $\vec v\perp\vec w$; and (ii) \textit{Completeness:} $\sum_{i=1}^{d}\mu({\bf v}_i)=1$ for every orthonormal basis $\{\vec v_i\}_{i=1}^{d}\subset\mathcal{K}$.
\end{definition}

\noindent Existence of such a set demonstrates the impossibility of assigning predetermined noncontextual outcomes to quantum measurements and thereby constitutes a logical proof of the KS theorem. The original construction of Kochen and Specker employed $117$ vectors in $\mathbb{C}^{3}$ \cite{Kochen1967}, and numerous smaller logical KS sets have subsequently been discovered in dimensions three and four \cite{Peres1991,Peres1993book,Bub1996,Penrose2000,Kernaghan1994,Cabello1996} (see also \cite{Budroni2022}). More generally, contextuality may also arise through state-independent or state-dependent contradictions that do not require the absence of KS colorings \cite{Yu2012,Bengtsson2012,Xu2015,Clifton1993,Klyachko2008}. In this work we focus exclusively on logical KS contextuality and introduce a mechanism through which logical KS contradictions can emerge from KS-colorable parent vector configurations.

\medskip
\noindent Our construction is based on a remarkable property of Hamilton's quaternion algebra \cite{Ebbinghaus1991}, which associates to every vector in $\mathbb{R}^{4}$ an orthogonal basis containing that vector. 

\begin{definition}\label{def2}
(Hamilton extension) For a vector $\Vec v=(a,b,c,d)\in\mathbb{R}^{4}$, the Hamilton extension of $\Vec v$, denoted by $\mathrm{H}[\Vec v]$, is the set $\mathrm{H}[\Vec v]=\left\{{\bf v}^{0},{\bf v}^{1},{\bf v}^{2},{\bf v}^{3}\right\}$, where ${\bf v}^{0}:=[(a,b,c,d)]$, ${\bf v}^{1}=[(-b,a,-d,c)]$, ${\bf v}^{2}=[(-c,d,a,-b)]$, ${\bf v}^{3}:=[(-d,-c,b,a)]$ are four mutually orthogonal rays. For a finite set $\mathcal V=\{\Vec v_i\}_{i=1}^{n}\subset\mathbb{R}^{4}$, its Hamilton extension is defined element-wise by $\mathrm{H}[\mathcal V]:=\bigcup_{i=1}^{n}\mathrm{H}[\Vec v_i]$.
\end{definition}

\noindent An explicit example illustrating this construction is shown in Fig.~\ref{fig1}. The following result establishes a key structural property of Hamilton extension (its proof is deferred to the Appendix).

\begin{proposition}\label{prop1}
For every vector set $\mathcal{V}\subseteq\mathbb{R}^{4}$, the Hamilton extension is idempotent, i.e. $\mathrm{H}\!\left[\mathrm{H}[\mathcal{V}]\right]=\mathrm{H}[\mathcal{V}]$.
\end{proposition}

\noindent To analyze contextuality generated by Hamilton extension, it is convenient to work with orthogonality graphs. A vector set $\mathcal{V}\subseteq\mathbb{R}^{4}$ determines an orthogonality graph $\mathtt{G}_{\mathcal{V}}=(\mathtt{V},\mathtt{E})$, whose vertices correspond to vectors in $\mathcal{V}$ and where two vertices are adjacent whenever the corresponding vectors are orthogonal. In this representation, logical KS contextuality is equivalent to the nonexistence of a valid vertex-coloring $\mu:\mathtt{V}\to\{0,1\}$ satisfying the exclusivity and completeness constraints of Definition~\ref{def1} (also called KS coloring). The Hamilton extension of $\mathtt{G}_{\mathcal{V}}$, denoted by $\mathtt{G}_{\mathrm{H}[\mathcal{V}]}$, is the orthogonality graph associated with the Hamilton-extended vector set $\mathrm{H}[\mathcal{V}]$.

\medskip
\noindent Two vectors $\Vec u,\Vec v\in\mathbb{R}^{4}$ are said to be Hamilton-distinct if $\mathrm{H}[\Vec u]\neq \mathrm{H}[\Vec v]$. Each Hamilton-distinct vector generates a distinct orthogonal basis under Definition~\ref{def2}, and therefore a distinct measurement context, which we call a \emph{Hamilton context}. Crucially, additional measurement contexts may arise from orthogonality relations among vectors belonging to different Hamilton contexts. We refer to these as \emph{emergent contexts}. They are of two types: (i) type (2\text{-}2), consisting of two extended children from each of two Hamilton-distinct parent vectors, and (ii) type (1\text{-}1\text{-}1\text{-}1), consisting of one extended child from each of four mutually Hamilton-distinct parent vectors. The central theme of this work is that emergent contexts can qualitatively change the contextual structure of a vector set. In particular, they can generate logical KS contradictions in Hamilton extensions of parent sets that themselves admit consistent KS colorings.

\begin{figure*}[t!]
\centering
\begin{tabular}{ccc}
\includegraphics[height=5cm]{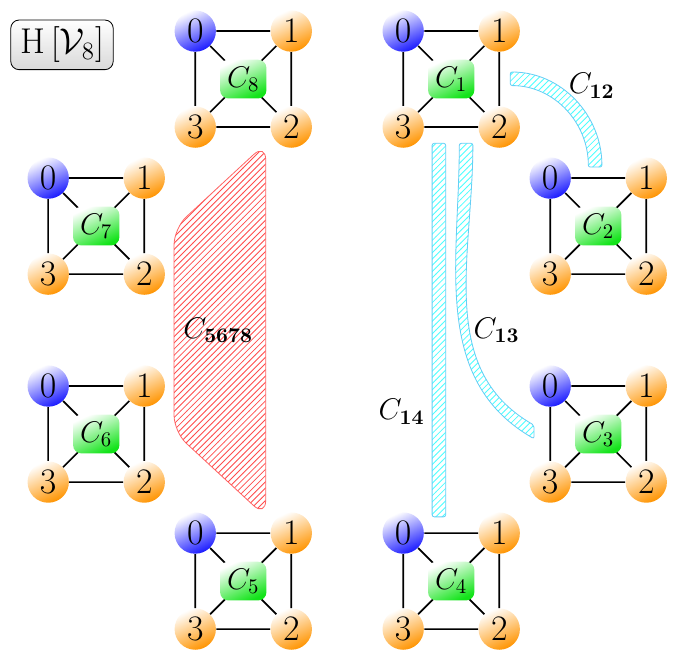}
\hspace{.5cm}
\includegraphics[height=5cm]{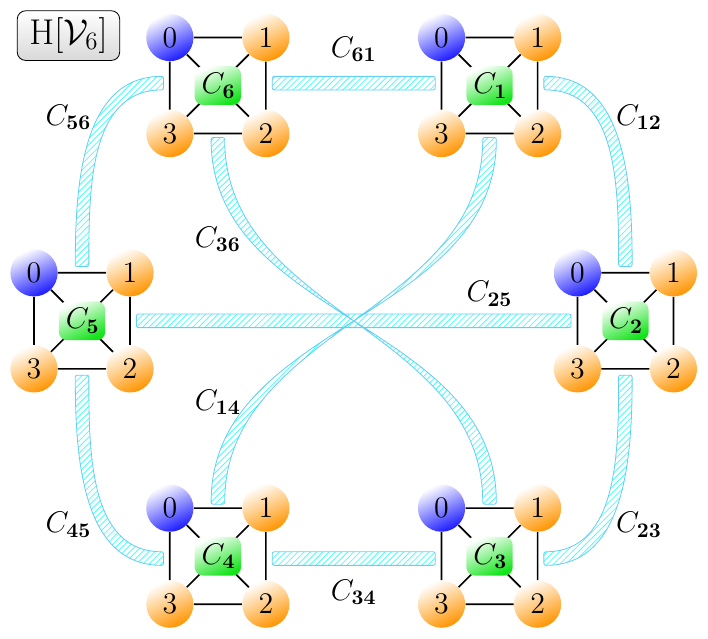}
\hspace{.5cm}
\includegraphics[height=5cm]{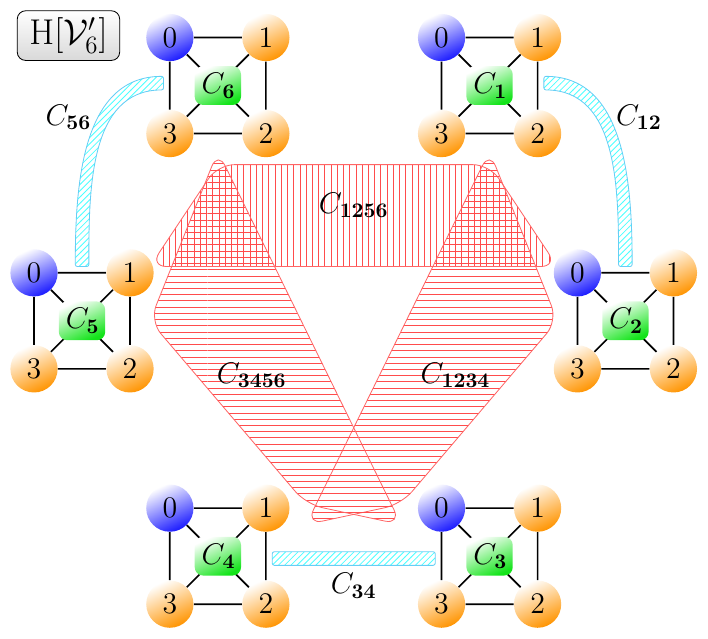}
\end{tabular}
\hrule\vspace{-.35cm}
\begin{align*}
&\hspace{.5cm}\scriptstyle\text{Hamilton contexts:}~C_i\equiv\mathrm{H}[\vec{u}_i]:=\left\{{\bf u}^0_i,{\bf u}^1_i,{\bf u}^2_i,{\bf u}^3_i\right\};~\text{Emergent contexts:}~C_{{\bf p}{\bf q}}^{\alpha\beta|\gamma\delta} :=\left\{{\bf u}_{\bf p}^{\alpha},{\bf u}_{\bf p}^{\beta},{\bf u}_{\bf q}^{\gamma},{\bf u}_{\bf q}^{\delta}\right\}~\big[\text{Type (2-2)}\big]~\text{and}~ C_{{\bf p}{\bf q}{\bf r}{\bf s}}^{\alpha|\beta|\gamma|\delta} :=\left\{{\bf u}_{\bf p}^{\alpha},{\bf u}_{\bf q}^{\beta},{\bf u}_{\bf r}^{\gamma},{\bf u}_{\bf s}^{\delta}\right\}~\big[\text{Type (1-1-1-1)}\big];\\
&\hspace{3cm}\scriptstyle\mathrm{H}\left[\mathcal{V}_8\right]:\left\{C^{01|01}_{\bf 12},~C^{23|23}_{\bf 12},~C^{02|02}_{\bf 13},~C^{13|13}_{\bf 13},~C^{03|03}_{\bf 14},~C^{12|12}_{\bf 14},~C^{1|0|2|3}_{\bf 5678},~C^{0|1|3|2}_{\bf 5678},~C^{2|3|1|0}_{\bf 5678},~C^{3|2|0|1}_{\bf 5678}\right\};\vspace{-.25cm}\\
&\scriptstyle\mathrm{H}\left[\mathcal{V}_6\right]:\left\{C^{03|03}_{\bf 12},~C^{12|12}_{\bf 12},~C^{02|13}_{\bf 23},~C^{13|02}_{\bf 23},~C^{03|03}_{\bf 34},~C^{12|12}_{\bf 34},~C^{02|13}_{\bf 45},~C^{13|02}_{\bf 45},~C^{03|03}_{\bf 56},~C^{12|12}_{\bf 56},~C^{02|13}_{\bf 61},~C^{13|02}_{\bf 61},~C^{01|23}_{\bf 14},~C^{23|01}_{\bf 14},~C^{01|23}_{\bf 25},~C^{23|01}_{\bf 25},~C^{01|23}_{\bf 36},~C^{23|01}_{\bf 36}\right\};\\\vspace{-.25cm}
&\hspace{1cm}\scriptstyle\mathrm{H}\left[\mathcal{V}^\prime_6\right]:\left\{C^{01|01}_{\bf pq},~C^{23|23}_{\bf pq},~C_{\bf 1234}^{\mu|\mu|\mu|\mu},C_{\bf 1256}^{0|1|0|0},C_{\bf 1256}^{1|0|1|1},C_{\bf 1256}^{2|3|2|2},C_{\bf 1256}^{3|2|3|3},C_{\bf 3456}^{0|1|3|2},C_{\bf 3456}^{1|0|2|3},C_{\bf 3456}^{2|3|1|0},C_{\bf 3456}^{3|2|0|1}~|~{\bf pq}={\bf 12},{\bf 34},{\bf 56}~~\&~~\mu=0,1,2,3\right\}.
\end{align*}
~\vspace{-.5cm}
\hrule\vspace{-.25cm}
\caption{(Color online) Hamilton extensions generating logical KS contextuality. Left: the extension $\mathrm{H}[\mathcal{V}_{8}]$ of Theorem~\ref{theo1}, yielding a new $32$-vector logical KS set in $\mathbb{R}^{4}$. Middle: the extension $\mathrm{H}[\mathcal{V}_{6}]$ of Theorem~\ref{theo2}, which coincides with the Peres-$24$ \cite{Peres1991}. Green, cyan, and red indicate Hamilton contexts, emergent contexts of Type $(2\text{-}2)$, and emergent contexts of Type $(1\text{-}1\text{-}1\text{-}1)$, respectively. The set $\mathtt{P}_{24}$ contains as subsets the KS constructions $\mathtt{K_{20}}$ \cite{Kernaghan1994} and $\mathtt{C_{18}}$ \cite{Cabello1996}. Right: 24-vector KS set $\mathrm{H}\left[\mathcal{V}^\prime_6\right]$. It contains 12 Type (1-1-1-1) and 6 Type (2-2) emergent contexts and structurally different than Peres-$24$ construction which contains 18 Type (2-2) emergent contexts; additional structural relations are discussed in the text and Appendix.} \vspace{-.25cm}
\label{fig2}
\end{figure*}

\medskip
\noindent {\it Emergence of logical KS contextuality.--} We begin with the $13$-vector construction $\mathtt{YO_{13}}$ due to Yu and Oh \cite{Yu2012}. Although this construction gives rise to a SI-C proof through an algebraic contradiction, its orthogonality graph nevertheless admits a valid KS coloring. The graph admits an orthogonal realization in $\mathbb{R}^{3}$. By adjoining an apex vertex $v_A$ adjacent to every vertex of $\mathtt{YO_{13}}$ one obtains the graph $\mathtt{YO_{14}}\equiv \mathtt{YO_{13}}+1$, which admits an orthogonal realization in $\mathbb{R}^{4}$ by embedding the vectors of $\mathtt{YO}_{13}$ into $\mathbb{R}^{4}$ and choosing the apex vector $\Vec v_A=(0,0,0,1)$. Remarkably, the apex operation weakens the contextual properties of the graph: as shown in Ref.~\cite{Cabello2015}, $\mathtt{YO_{14}}$ no longer exhibits SI-C and instead admits only a SD-C proof. The following theorem shows that Hamilton extension not only recovers the contextuality lost under apex augmentation but does so in its strongest possible form.

\begin{theorem}\label{theo1}
The graph $\mathrm{H}\!\left[\mathtt{YO_{14}}\right]$ does not admit a KS coloring and therefore exhibits logical KS contextuality.
\end{theorem}

\begin{proof}
(Sketch) The vector set underlying $\mathtt{YO_{14}}$ contains a subset 
\begin{align*}
\scriptsize\mathcal{V}_8=\left\{
\begin{aligned}
&\vec u_1=(0,0,0,1);~~\vec u_2=(1,-1,0,0);\vec u_3=(1,0,-1,0);\\
&\vec u_4=(0,1,-1,0);\vec u_5=(1,-1,1,0);\vec u_6=(-1,1,1,0);\\
&\hspace{1.5cm}\vec u_7=(1,1,-1,0);~\vec u_8=(1,1,1,0)
\end{aligned}
\right\}.
\end{align*}
of cardinality eight, wherein the vectors are pairwise Hamilton-distinct and $\mathrm{H}\!\left[\mathtt{YO_{14}}\right]=\mathrm{H}\!\left[\mathcal{V}_8\right]$. Consequently, $\mathrm{H}\!\left[\mathtt{YO_{14}}\right]$ consists of $32$ vectors arranged into eight Hamilton contexts. In addition, the Hamilton extension generates ten emergent contexts: six are of type $(2\text{-}2)$ and four are of type $(1\text{-}1\text{-}1\text{-}1)$ (see Fig.~\ref{fig2}). A detailed analysis of these Hamilton and emergent contexts, presented in the Appendix, shows that no KS coloring of $\mathrm{H}\!\left[\mathtt{YO_{14}}\right]$ exists, completing the proof.
\end{proof}

\noindent Notably, the set $\mathrm{H}\!\left[\mathcal{V}_8\right]$ gives rise to a previously unreported $32$-vector logical KS set in $\mathbb{R}^{4}$. The choice of the subset $\mathcal{V}_8\subset\mathtt{YO}_{14}$ is not unique: there exist $32$ distinct selections of pairwise Hamilton-distinct vectors that generate the same Hamilton extension.

\medskip
\noindent The Hamilton extension defined above applies to vector sets in $\mathbb{R}^{4}$. It admits a natural generalization to vector sets in $\mathbb{C}^{4}$. Let $\mathcal{V}=\{\Vec{v}_i\}_{i=1}^{n}=\{\Vec{v}^{\,r}_i\}_{i=1}^{k}\cup\{\Vec{v}^{\,c}_i\}_{i=k+1}^{n}\subset\mathbb{C}^{4}$, where the vectors $\{\Vec{v}^{\,r}_i\}_{i=1}^{k}$ have only real coefficients and the vectors $\{\Vec{v}^{\,c}_i\}_{i=k+1}^{n}$ possess genuinely complex coefficients. We define the generalized Hamilton extension by $\mathrm{H}_{g}[\mathcal{V}]:=\bigcup_{i=1}^{k}\mathrm{H}[\Vec{v}^{\,r}_i]\cup\{\Vec{v}^{\,c}_j\}_{j=k+1}^{n}$, which is nontrivial whenever $\mathcal V$ contains at least one real vector. The logical KS contradiction generated by $\mathcal V_8$ then immediately extends to a broad class of complex vector configurations.

\begin{corollary}\label{coro1}
Let $\mathcal{V}\subset\mathbb{C}^{4}$ contain $\mathcal{V}_8$ as a subset. Then the generalized Hamilton extension $\mathrm{H}_{g}[\mathcal{V}]$ is logically KS contextual.
\end{corollary}

\noindent Since $\mathrm{H}\!\left[\mathcal{V}_8\right]\subseteq\mathrm{H}_{g}\!\left[\mathcal{V}\right]$, the claim follows from the hereditary property of KS non-colorability: every super-set of a KS non-colorable set is itself KS non-colorable. This corollary reveals an unexpected connection between otherwise unrelated constructions. For example, the $18$-vector set $\mathtt{CN}_{18}\subset\mathbb{R}^{4}$ introduced by Cameron and Newman \textit{et al.} in the context of graph coloring game \cite{Cameron2007} contains $\mathcal V_8$ as a subset. Consequently, its Hamilton extension yields a $36$-vector logical KS set, despite originating from a framework entirely unrelated to $\mathtt{YO}_{14}$ (a more detailed exposition is presented in the Appendix).

\medskip
\noindent {\it Recovery of logical KS contextuality.--} Theorem~\ref{theo1} demonstrates both the emergence of logical KS contextuality from a KS-colorable parent configuration and the recovery of contextuality lost under apex augmentation. We now show that the recovery phenomenon is considerably more general. Consider a KS graph $\mathtt{G}\equiv(\mathtt{V},\mathtt{E})$. Construct a new graph $\mathtt{G}+1$ by adjoining an apex vertex $v_A$ adjacent to every vertex of $\mathtt{G}$, i.e., $v_A\sim v,~ \forall\, v\in\mathtt{V}$. The graph $\mathtt{G}+1$ always admits a valid KS coloring: assigning the value $1$ to $v_A$ and the value $0$ to all remaining vertices satisfies the exclusivity and completeness constraints. Consequently, apex augmentation destroys the original logical KS contradiction and renders $\mathtt{G}+1$ KS colorable. Now, given a vector $\Tilde{v}=(a,b,c)\in\mathbb{C}^{3}$, define its embedding in $\mathbb{C}^{4}$ by $\mathtt{Emb}:\Tilde{v}\mapsto \mathtt{E}(\Tilde{v})=(a,b,c,0)$. For a vector set $\mathcal{V}=\{\Tilde{v}_j\}\subset\mathbb{C}^{3}$, let
$\mathtt{E}(\mathcal{V})=\{\mathtt{E}(\Tilde{v}_j)\}\subset\mathbb{C}^{4}$. Let $\mathcal{V}_7=\{\Tilde{v}_i\}_{i=2}^{8}\subset\mathbb{R}^{3}$ be such that
$\mathtt{E}(\mathcal{V}_7)=\mathcal{V}_8\setminus\{\Vec u_1\}$. With this we have the following general result

\begin{corollary}\label{coro2}
Let $\mathcal{V}\subset\mathbb{C}^{3}$ contain $\mathcal{V}_7$ as a subset and let $\Vec v_A=(0,0,0,1)$. Then the generalized Hamilton extension of the KS-colorable set $\mathtt{E}(\mathcal{V})\cup\{\Vec v_A\}$ exhibits logically KS contextuality.
\end{corollary}

\noindent The proof follows immediately from the hereditary nature of KS non-colorability. This recovery phenomenon is exhibited by several well-known KS constructions.

\begin{example*}
For each of the logical KS constructions $\mathtt{P}_{33}$ of Peres \cite{Peres1991}, $\mathtt{CK}_{31}$ of Conway and Kochen \cite{Peres1993book}, and $\mathtt{B}_{33}$ of Bub \cite{Bub1996}, apex augmentation destroys logical KS contextuality by rendering $\mathtt{G}+1$ KS colorable, whereas its Hamilton extension $\mathrm{H}[\mathtt{G}+1]$ is KS non-colorable and therefore exhibits logical KS contextuality.
\end{example*}

\noindent Indeed, each of the vector sets $\mathtt{P}_{33}$, $\mathtt{CK}_{31}$, and $\mathtt{B}_{33}$ contains $\mathtt{YO}_{13}$ and therefore contains $\mathcal{V}_7$ as a subset \cite{Pavicic2023,Cabello2025,Pavicic2025}. Consequently, Corollary~\ref{coro2} implies that $\mathrm{H}_g[\mathtt{G}+1]$ is KS non-colorable for every $\mathtt{G}\in\left\{\mathtt{P}_{33},\mathtt{CK}_{31},\mathtt{B}_{33}\right\}$. Hence Hamilton extension restores the logical KS contextuality destroyed by apex augmentation.

\medskip
\noindent {\it Optimality of the parent-set cardinality.--} A natural question is whether the eight-vector set $\mathcal{V}_8$ appearing in Theorem~\ref{theo1} is minimal among parent sets whose Hamilton extensions exhibit logical KS contextuality. The following result shows that considerably smaller parent sets suffice.

\begin{theorem}\label{theo2}
There exists a six-vector set 
\begin{align*}
\scriptsize\mathcal{V}_6:=\left\{\begin{aligned}
&\vec u_1=(1,0,0,0);~~~~\vec u_2=(0,1,1,0);~~~~\vec u_3=(1,1,1,1);\\
&\vec u_4=(1,-1,0,0);~\vec u_5=(-1,1,1,1);~\vec u_6=(1,0,1,0)   
\end{aligned}\right\}
\end{align*}
whose Hamilton extension $\mathrm{H}[\mathcal{V}_6]$ exhibits logical KS contextuality.
\end{theorem}

\noindent The proof follows from the observation that $\mathrm{H}[\mathcal{V}_6]$ coincide with Peres-$24$ construction $\mathtt{P_{24}}$, a well-known logical KS set \cite{Peres1991,Cubitt2010}. Notably, the choice of $\mathcal{V}_6$ is not unique: in fact, there exist $4096$ distinct six-vector parent sets whose Hamilton extensions coincide with $\mathtt{P_{24}}$. Interestingly, we now present a new construction of a 24-vector KS set.
\begin{construction}\label{const1}
The following set of six vectors 
\begin{align*}
\scriptsize\mathcal{V}^\prime_6:=\left\{\begin{aligned}
&\vec u_1=(1,0,0,0);~~~~\vec u_2=(0,0,1,1);~~~~\vec u_3=(0,\sqrt{2},1,-1);\\
&\vec u_4=(0,\sqrt{2},-1,1);~\vec u_5=(0,\sqrt{2},1,1);~\vec u_6=(0,-\sqrt{2},1,1) 
\end{aligned}\right\}
\end{align*}
consists of mutually Hamilton-distinct vectors. The Hamilton extension of the set $\mathrm{H}[\mathcal{V}^\prime_6]$ generates a 24-vector KS set that is not isomorphic to $\mathtt{P_{24}}$, thereby providing a new KS construction of the same cardinality (see Fig.~\ref{fig2}).
\end{construction}
\noindent A proof of this claim is presented in the Appendix. The examples $\mathcal{V}^\prime_6$ and $\mathcal{V}_6$ naturally motivate the question of whether five parent vectors are already sufficient to generate logical KS contextuality through Hamilton extension. The following theorem shows that this is not the case.

\begin{theorem}\label{theo3}
For every five-vector set $\mathcal{V}_5\subset\mathbb{R}^{4}$, the Hamilton extension $\mathrm{H}[\mathcal{V}_5]$ admits a KS coloring.
\end{theorem}

\noindent The proof is given in the Appendix.  Theorems~\ref{theo2} and \ref{theo3} establish a sharp threshold for the emergence of logical KS contextuality under Hamilton extension. While every five-vector parent set yields a KS-colorable Hamilton extension, a suitably chosen six-vector parent set already generates the Peres-$24$ construction. Combined with the result of Xu \textit{et al.} \cite{Xu2020} that every KS set contains at least $18$ vectors (confirming a conjecture of Peres \cite{Peres2003}), this identifies six as the minimum parent-set cardinality capable of generating logical KS contextuality through Hamilton extension.

\noindent Several additional structural features are worth noting. First, the Peres-$24$ set is itself fixed under Hamilton extension: $\mathrm{H}\left[\mathtt{P_{24}}\right]=\mathtt{P_{24}}$. Second, both the $20$-vector set $\mathtt{K}_{20}$ of Kernaghan \cite{Kernaghan1994} and the $18$-vector set $\mathtt{C}_{18}$ of Cabello \textit{et al.} \cite{Cabello1996} occur as subsets of $\mathrm{H}\left[\mathcal{V}_6\right]$. Finally, the parent sets $\mathcal{V}_6$ and $\mathcal{V}'_6$ are irreducible in the sense that removing any vector yields a five-vector set whose Hamilton extension is KS-colorable by Theorem~\ref{theo3}. In contrast, the set $\mathcal{V}_8$ (Theorem~\ref{theo1}) is only weakly irreducible. 
\begin{construction}\label{const2}
The Hamilton extension $\mathrm{H}[\mathcal{V}_8\setminus\{\vec u\}]$ is KS-uncolorable for $\vec u\in\{\vec u_2,\vec u_3,\vec u_4\}$, and thus provide examples of irreducible seven-vector parent sets. However, for $\vec u\in\{\vec u_1,\vec u_5,\vec u_6,\vec u_7,\vec u_8\}$ the extension $\mathrm{H}[\mathcal{V}_8\setminus\{\vec u\}]$ is KS-colorable. Thus, $\mathcal{V}_8$ is not irreducible in the strict sense, but only weakly irreducible.    
\end{construction}

\noindent {\it Discussion.--} We have introduced the Hamilton extension—a constructive transformation that associates a measurement context with each real vector of a four-dimensional parent set while simultaneously generating additional contexts through orthogonality relations among the extended vectors. These emergent compatibility relations fundamentally reshape the measurement structure, enabling logical Kochen--Specker (KS) contextuality to arise from configurations that are themselves KS colorable and restoring logical contradictions eliminated by apex-vertex augmentation. Hamilton extension is therefore not merely a dimensional embedding but a genuinely nontrivial structural transformation capable of generating new contextual phenomena.

\noindent Our results also reveal a sharp threshold governing this emergence mechanism. While the Hamilton extension of every five-vector parent set remains KS colorable, suitably chosen six-vector parent sets already generate logical KS contextuality. Thus, six is the minimum parent-set cardinality for which contextuality can emerge through Hamilton extension. This observation suggests an alternative perspective on KS constructions: instead of classifying them solely by the size of the resulting KS set, one may classify them by the complexity of the underlying parent configuration and the transformation that generates the contradiction. In this sense, Hamilton extension provides a new organizing principle based on generative structures rather than final configurations.

\noindent More broadly, our work points toward a new paradigm in which contextuality is understood as an emergent property of systematically transformed compatibility structures. It raises several natural directions for future investigation, including higher-dimensional extension procedures based on other algebraic or orthogonal-design frameworks, graph-theoretic criteria characterizing when emergent compatibility relations produce logical KS contradictions, and a classification of minimal irreducible parent sets. The ability to systematically generate compact KS sets may also provide new building blocks for contextuality-based quantum information protocols, including self-testing, device certification, and randomness generation, where smaller and structurally controlled contextual configurations can offer practical advantages. More generally, Hamilton extension provides a new framework for exploring the interplay between algebraic constructions, graph-theoretic characterizations of contextuality, and their operational applications in quantum information.

\medskip
\noindent{\bf Acknowledgement}: We thankfully acknowledge useful discussions with Kunika Agarwal and Subhendu B Ghosh. JK, BC, and SRC acknowledge support from University Grants Commission, India (Reference no. 231620167137, 241620129062, and 211610113404, respectively). MB acknowledges the financial support through the National Quantum Mission (NQM) of the Department of Science and Technology, Government of India.


%

\newpage
\onecolumngrid

\section{Appendix A: Hamilton Extension---Mathematical Structures and Measurement Contexts}

\noindent Before presenting our proofs, in this section, we will first analyze several mathematical structures of Hamilton extension that will be useful in our subsequent arguments.

\subsection{A-1. Hamilton Extension of Projectors}  

\noindent Since we are dealing with Kochen-Specker contextuality, often it is customary to work in the language of projectors instead of vectors. Let $\mathbb{RP}^3$ denote the real projective space of rays in $\mathbb{R}^4$. For a nonzero vector $\vec v\in\mathbb{R}^4$, let $[\vec v]$ denote the ray generated by $\vec v$. The Hamilton extension of a ray $[\vec v]\in\mathbb{RP}^3$ is the ordered set
\begin{align}
\mathrm H\big[[\vec v]\big]:=\big\{[\vec v]^\alpha=[\mathrm e_\alpha\theta(\vec v)]:\alpha=0,1,2,3\big\}.\end{align}
This definition is independent of the choice of representative $\vec v$. For notational convenience, we subsequently write $[\vec v]=\bf v$, and thus above equation reads as
\begin{align}
\mathrm H[{\bf v}]:=\big\{{\bf v}^\alpha=[\mathrm e_\alpha\theta(\vec v)]:\alpha=0,1,2,3\big\}.   
\end{align}

\subsection{A-2. Quaternion Representation} 
\noindent For any vector $\vec v=(a,b,c,d)\in\mathbb{R}^{4}$, define the linear isomorphism $\theta:\mathbb{R}^{4}\rightarrow\mathbb{H}$ by 
\begin{align}
(a,b,c,d)\simeq \theta(\vec v):=
a+b~\mathrm e_{1}+c~\mathrm e_{2}+d~\mathrm e_{3},    
\end{align}
where $\mathbb H$ denotes the quaternion algebra generated by
$\{\mathrm e_{1},\mathrm e_{2},\mathrm e_{3}\}$ satisfying
\begin{align}\label{quatmulrul}
\mathrm e_i\mathrm e_j=-\delta_{ij}+\epsilon_{ijk}\mathrm e_k .    
\end{align}
Here, $\delta_{ij}$ and $\epsilon_{ijk}$ denote the Kronecker delta and Levi--Civita symbol, respectively. Under this isomorphism, the Hamilton extension of a vector $\vec v$ can be written compactly as
\begin{align}
\mathrm H[\vec v]\simeq\big\{\mathrm e_\alpha\theta(\vec v)\,:\,\alpha=0,1,2,3\big\},  
\end{align}
where, and $\mathrm e_0:=1$. Namely, 
\begin{align}
\left.\begin{aligned}
&\mathrm e_0\theta(\vec v)=~~a+b~\mathrm e_1+c~\mathrm e_2+d~\mathrm e_3\simeq (a,b,c,d):=\vec v^0,\\ 
&\mathrm e_1\theta(\vec v)=-b+a~\mathrm e_1-d~\mathrm e_2+c~\mathrm e_3\simeq (-b,a,-d,c):=\vec v^1,\\ 
&\mathrm e_2\theta(\vec v)=-c+d~\mathrm e_1+a~\mathrm e_2-b~\mathrm e_3\simeq (-c,d,a,-b):=\vec v^2,\\
&\mathrm e_3\theta(\vec v)=-d-c~\mathrm e_1+b~\mathrm e_2+a~\mathrm e_3\simeq (-d,-c,b,a):=\vec v^3.
\end{aligned}\right\}.
\end{align}
From now on we will use these different representations (vector, projector, and quaternion) interchangeably.

\begin{lemma}\label{t3l1}
Let ${\bf u},{\bf v}\in\mathbb{RP}^3$. If ${\bf u}\in \mathrm H[{\bf v}]$, then $\mathrm H[{\bf u}]=\mathrm H[{\bf v}]$.
\end{lemma}
\begin{proof}
Since ${\bf u}\in \mathrm H[{\bf v}]$, there exist
$\beta\in\{0,1,2,3\}$ and $\lambda\neq0$ such that $\theta(\vec u)=\lambda\,\mathrm e_\beta\theta(\vec v)$, where $\vec u$ and $\vec v$ are representatives of $u$ and $v$, respectively (throughout the paper we will consider the representative vectors to be normalized). Hence
\begin{align}
\mathrm H[{\bf u}]&=\left\{
[\mathrm e_\alpha\theta(\vec u)]:\alpha=0,1,2,3\right\}=\left\{[\mathrm e_\alpha\mathrm e_\beta\theta(\vec v)]:\alpha=0,1,2,3\right\}.    
\end{align}
By the quaternion multiplication rules, for every fixed $\beta\in\{0,1,2,3\}$ and every $\alpha\in\{0,1,2,3\}$ there exist $\gamma\in\{0,1,2,3\}$ and $\mathrm{s}\in\{\pm1\}$ such that $\mathrm e_\alpha\mathrm e_\beta=\mathrm s\,\mathrm e_\gamma$. Moreover, the map $\alpha\mapsto\gamma$ is a permutation of
$\{0,1,2,3\}$. Since $[\mathrm s\,\mathrm e_\gamma\theta(\vec v_2)]=[\mathrm e_\gamma\theta(\vec v_2)]$, it follows that
\begin{align}
\mathrm H[{\bf u}]&=\left\{[\mathrm e_\alpha\mathrm e_\beta\theta(\vec v)]:\alpha=0,1,2,3\right\}=\left\{[\mathrm e_\gamma\theta(\vec v)]:\gamma=0,1,2,3\right\}=\mathrm H[{\bf v}].   
\end{align}
This completes the proof. 
\end{proof}

\begin{definition}
Two rays ${\bf u},{\bf v}\in\mathbb{RP}^3$ are said to be
\emph{Hamilton-distinct} if $\mathrm{H}[{\bf u}]\neq \mathrm{H}[{\bf v}]$; equivalently, ${\bf u}\notin \mathrm{H}[{\bf v}]$ or, equivalently, ${\bf v}\notin \mathrm{H}[{\bf u}]$. 
\end{definition}
\noindent Analogous definition applies to vectors in $\mathbb{R}^4$. 

\begin{lemma}\label{t3l2}
Let ${\bf u},{\bf v}\in\mathbb{RP}^3$, with $\mathrm H[{\bf u}]=\left\{{\bf u}^\alpha\right\}_{\alpha=0}^3$ and $\mathrm H[{\bf v}]=\left\{{\bf v}^\beta\right\}_{\beta=0}^3$.Then, $\Tr[{\bf u}^\alpha {\bf v}^\beta]=\Tr\left[{\bf u}^{\alpha+\mu(\bmod4)}{\bf v}^{\beta\mp\mu(\bmod4)}\right]$; the `{+ve}' sign is when $|\alpha - \beta|\in\{0,2\}$ and `{-ve}' sign is when $|\alpha - \beta|\in\{1,3\}$ (see also Fig.~\ref{figs1}).
\end{lemma}

\begin{proof}
Let $\Vec u =(u_0,u_1,u_2,u_3), \Vec v=(v_0,v_1,v_2,v_3) \in \mathbb{R}^4$ be the normalized vector representatives of ${\bf u}$ and ${\bf v}$, respectively. Then
\begin{subequations}
\begin{align}
\mathrm{H}[{\bf u}]&=\left\{\begin{aligned}
&{\bf u}^0=[(u_0,u_1,u_2,u_3)],~~~~{\bf u}^1=[(-u_1,u_0,-u_3,u_2)],\\
&{\bf u}^2=[(-u_2,u_3,u_0,-u_1)],{\bf u}^3=[(-u_3,-u_2,u_1,u_0)]
\end{aligned}
\right\},\\
\mathrm{H}[{\bf v}]&=\left\{\begin{aligned}
&{\bf v}^0=[(v_0,v_1,v_2,v_3)],~~~~~{\bf v}^1=[(-v_1,v_0,-v_3,v_2)],\\
&{\bf v}^2=[(-v_2,v_3,v_0,-v_1)],~{\bf v}^3=[(-v_3,-v_2,v_1,v_0)]
\end{aligned}
\right\}.
\end{align}
\end{subequations}

\noindent A straightforward calculation yields the inner product between pairs of vectors $(\Vec u^\mu,\Vec v^\nu)$ as listed in the Table~\ref{tabs1}.
\begin{table}[h!]
\centering
\begin{tabular}{|c|c|}
\hline
$(\Vec u^\mu,\Vec v^\nu)$& $|\langle\Vec u^\mu,\Vec v^\nu\rangle|=\Tr[\bf u^\mu \bf v^\nu]$\\\hline

$(\Vec u^0,\Vec v^0)$;$(\Vec u^1,\Vec v^1)$;$(\Vec u^2,\Vec v^2)$;$(\Vec u^3,\Vec v^3)$ & ~~~~~~~$|u_0v_0+u_1v_1+u_2v_2+u_3v_3|$~~~~~~~\\\hline

$(\Vec u^0,\Vec v^1)$;$(\Vec u^1,\Vec v^0)$;$(\Vec u^2,\Vec v^3)$;$(\Vec u^3,\Vec v^2)$ & $|-u_0v_1+u_1v_0-u_2v_3+u_3v_2|$\\\hline

$(\Vec u^0,\Vec v^2)$;$(\Vec u^1,\Vec v^3)$;$(\Vec u^2,\Vec v^0)$;$(\Vec u^3,\Vec v^1)$ & $|-u_0v_2+u_1v_3+u_2v_0-u_3v_1|$\\\hline

$(\Vec u^0,\Vec v^3)$;$(\Vec u^1,\Vec v^2)$;$(\Vec u^2,\Vec v^1)$;$(\Vec u^3,\Vec v^0)$ & $|-u_0v_3-u_1v_2+u_2v_1+u_3v_0|$ \\\hline
\end{tabular}
\caption{Inner-product between the elements of $\mathrm{H}[{\bf u}]$ and $\mathrm{H}[{\bf v}]$ for ${\bf u},{\bf v}\in\mathbb{RP}^3$.}\label{tabs1}
\label{2-2-three}
\end{table}

\noindent A compact representation of all these inner products values is precisely the claim of present lemma. 
\end{proof}

\begin{figure}[t!]
\centering
\includegraphics[width=0.45\textwidth]{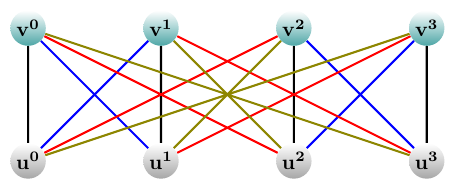} 
\caption{(Color online) Graphical representation of the inner-product relations between the elements of $\mathrm{H}[{\bf u}]$ and $\mathrm{H}[{\bf v}]$ for ${\bf u},{\bf v}\in\mathbb{RP}^3$. Each edge color labels a distinct inner-product value; edges sharing the same color represent equal inner products.}\vspace{-.25cm}
\label{figs1}
\end{figure}

\subsection{A-3. Hamilton context and Emergent context} 

\noindent Let ${\bf V}\subset \mathbb{RP}^3$, and let $\mathcal{C}({\bf V})$ denote the set of all measurement contexts of ${\bf V}$, namely, all orthonormal bases of $\mathbb{R}^4$ whose constituent rays belong to ${\bf V}$. For a ray ${\bf v}$, its Hamilton extension $\mathrm{H}[{\bf v}]$ forms a measurement context, denoted $C_{\bf v}$, since the four rays in $\mathrm{H}[{\bf v}]$ are mutually orthogonal.

\begin{definition}
A context $C_{\bf v}\in\mathcal{C}({\bf V})$ is called a \emph{Hamilton context} if there exists a ray ${\bf v}\in {\bf V}$ such that $C_{\bf v}=\mathrm{H}[{\bf v}]$, otherwise the context is called an emergent context. 
\end{definition}

\noindent Let $C=\{{\bf u}_1,{\bf u}_2,{\bf u}_3,{\bf u}_4\}$ be an emergent context. Since $C$ is not a Hamilton context, it is not contained in any single Hamilton extension. Let $\{{\bf v}_1,\ldots,{\bf v}_k\}$ be the minimal collection of pairwise Hamilton-distinct rays such that
\begin{align}
C\subseteq \bigcup_{i=1}^{k}\mathrm{H}[{\bf v}_i].    
\end{align}
The type of $C$ is defined by the partition of its four rays among the Hamilton extensions $\mathrm{H}[{\bf v}_1],\ldots,\mathrm{H}[{\bf v}_k]$. Since $k\ge 2$, the possible types are precisely the nontrivial partitions of $4$:
\begin{enumerate}
[itemsep=-.15cm, parsep=.15cm, topsep=2pt, leftmargin=.8cm]
\item[(i)] (1\text{-}1\text{-}1\text{-}1) type: ${\bf u}_i\in \mathrm{H}[{\bf v}_i],~i=1,2,3,4$; where $\{{\bf v}_1,{\bf v}_2,{\bf v}_3,{\bf v}_4\}$ are pairwise Hamilton-distinct.

\item[(ii)] (2-1-1) type: ${\bf u}_1\in \mathrm{H}[{\bf v}_1],~ {\bf u}_2\in \mathrm{H}[{\bf v}_2],\quad {\bf u}_3,{\bf u}_4\in \mathrm{H}[{\bf v}_3]$; where $\{{\bf v}_1,{\bf v}_2,{\bf v}_3\}$ are pairwise Hamilton-distinct.

\item[(iii)] (2\text{-}2) type: ${\bf u}_1,{\bf u}_2\in \mathrm{H}[{\bf v}_1],~ {\bf u}_3,{\bf u}_4\in \mathrm{H}[{\bf v}_2]$; where ${\bf v}_1$ and ${\bf v}_2$ are Hamilton-distinct.

\item[(iv)] (3\text{-}1) type: ${\bf u}_1,{\bf u}_2,{\bf u}_3\in \mathrm{H}[{\bf v}_1],~ {\bf u}_4\in \mathrm{H}[{\bf v}_2]$; where ${\bf v}_1$ and ${\bf v}_2$ are Hamilton-distinct.
\end{enumerate}
\noindent We will now prove that some of these types are not possible. 

\begin{lemma}\label{t3l3}
Let ${\bf u},{\bf v},{\bf w}\in\mathbb{RP}^3$ be Hamilton-distinct. There exists no emergent context of type (2-1-1) between $\mathrm{H}[{\bf u}]$, $\mathrm{H}[{\bf v}]$, and $\mathrm{H}[{\bf w}]$.
\end{lemma}
\begin{proof}
Using the isomorphism $\theta:\mathbb{R}^4\to\mathbb{H}$, the Hilbert--Schmidt inner product between two projectors ${\bf u},{\bf v}\in\mathbb{RP}^3$ turns out to be
\begin{align}\label{t3l2e1}
\tr[{\bf uv}]=\left(\operatorname{Re}[\theta(\vec u)\overline{\theta(\vec v)}]\right)^2,
\end{align}
where $\Vec u$ and $\Vec v$ are the respective vector (normalized) representatives of ${\bf u}$ and ${\bf v}$, and for a $q=\mathrm e_0 a_0+\sum_{i=1}^3\mathrm e_i a_i\in\mathbb{H}$ its quaternion conjugate $\bar q=\mathrm e_0 a_0-\sum_{i=1}^3\mathrm e_i a_i$. We thus have
\begin{align}\label{t3l2e2}
\Tr[{\bf u}^\alpha {\bf v}^\beta]=\left(\operatorname{Re}[\mathrm{e}_{\alpha}\mathrm{e}_{\beta}\theta(\vec u)\overline{\theta(\vec v)}]\right)^2.
\end{align}
Consider an emergent context $C=\{{\bf u}_1,{\bf u}_2,{\bf u}_3,{\bf u}_4\}$ of type (2-1-1), where ${\bf u}_1\in \mathrm{H}[{\bf v}_1],~ {\bf u}_2\in \mathrm{H}[{\bf v}_2],\quad {\bf u}_3,{\bf u}_4\in \mathrm{H}[{\bf v}_3]$ with $\{{\bf v}_1,{\bf v}_2,{\bf v}_3\}$ being pairwise Hamilton-distinct. Without any loss of generality, we can consider ${\bf u}_1={\bf v}_1$, and $\theta(\vec v_1)=1$. Furthermore, consider ${\bf u}_3={\bf v}^\alpha_3$ and ${\bf u}_4={\bf v}^\beta_3$, where $\beta \ne \alpha$. Therefore, from Eq. \eqref{t3l2e2}, we have,
\begin{subequations}
\begin{align}
&\operatorname{Re}(\mathrm{e}_{\alpha}\theta(\vec v_3)\theta(\vec v_1))=0 \quad \Rightarrow\quad \operatorname{Re}(\mathrm{e}_{\alpha}\theta(\vec v_3))=0, \\
& \operatorname{Re}(\mathrm{e}_{\beta}\theta(\vec v_3)\theta(\vec v_1))=0 \quad \Rightarrow\quad \operatorname{Re}(\mathrm{e}_{\beta}\theta(\vec v_3))=0.
\end{align}  
\end{subequations}
We thus have
\begin{align}\label{t3l2e3}
\theta(\vec v_3) = \sum_{\gamma\in\{0,1,2,3\}\setminus\{\alpha,\beta\}}\mathrm{e}_{\gamma}v_3^{\gamma}
\end{align}
where $v_3^{\gamma}\in \mathbb{R}$ for all $\gamma$. Now, from the fact that ${\bf u}_2 \perp {\bf u}_3$ and ${\bf u}_2 \perp {\bf u}_4$, and from Eq.\eqref{t3l2e3}, we get 
\begin{subequations}\label{t3l2se4}
\begin{align}
\operatorname{Re}(\mathrm{e}_{\alpha}\theta(\vec v_3)\overline{\theta(\vec u_2)})=0, \\
\operatorname{Re}(\mathrm{e}_{\beta}\theta(\vec v_3)\overline{\theta(\vec u_2)})=0.
\end{align}
\end{subequations}
Therefore, $\theta(\vec v_3)\overline{\theta(\vec u_2)}$ is a quaternion number of the form:

\begin{align}\label{t3l2e4}
\theta(\vec v_3)\overline{\theta(\vec u_2)} = \sum_{\gamma\in\{0,1,2,3\}\setminus\{\alpha,\beta\}}\mathrm{e}_{\gamma}q^{\gamma}
\end{align}
Where, $q^{\gamma}\in \mathbb{R}$ for all $\gamma$.

Similarly, using ${\bf u}_2 \perp {\bf v}_1$, we get,
\begin{align}\label{t3l2e5}
\operatorname{Re}(\theta(\vec u_2)\theta(\vec v_1))=0 \quad \Rightarrow\quad \operatorname{Re}(\theta(\vec u_2))=0
\end{align}
Therefore, using Eq. \eqref{t3l2e3},\eqref{t3l2e4} and \eqref{t3l2e5}, and using the quaternion multiplication rule (see Eq. \eqref{quatmulrul}), we get, $\theta(\vec u_2)=\mathrm{e}_{\eta}$, where $\eta$ is the unique element of $\{0,1,2,3\}\setminus\{0,\alpha,\beta\}$. But then ${\bf u}_2\in \mathrm{H}[{\bf v}_1]$, which contradicts our assumption and hence completes the proof.
\end{proof}

\begin{lemma}\label{t3l4}
Let ${\bf u},{\bf v}\in\mathbb{RP}^3$ be Hamilton-distinct. There exists no emergent context of type (3-1) between $\mathrm{H}[{\bf u}]$ and $\mathrm{H}[{\bf v}]$.
\end{lemma}
\begin{proof}
Suppose, for contradiction, that $C$ is an emergent context of type (3-1). Then there exist Hamilton-distinct rays ${\bf v}_1,{\bf v}_2\in\mathbb{RP}^3$ such that
\begin{align}
|C\cap \mathrm{H}[{\bf v}_1]|=3
\quad\text{and}\quad
|C\cap \mathrm{H}[{\bf v}_2]|=1.    
\end{align}
Let $C\cap \mathrm{H}[{\bf v}_1]=\{{\bf r}_1,{\bf r}_2,{\bf r}_3\}$. Since $\mathrm{H}[{\bf v}_1]$ is a context, there exists a unique ray ${\bf r}_4\in \mathrm{H}[{\bf v}_1]$ such that $\mathrm{H}[{\bf v}_1]=\{{\bf r}_1,{\bf r}_2,{\bf r}_3,{\bf r}_4\}$. Because $C$ is also a context and already contains the three mutually orthogonal rays ${\bf r}_1,{\bf r}_2,{\bf r}_3$, its fourth ray must be orthogonal to each of them and must coincide with ${\bf r}_4$. Therefore, $C=\{{\bf r}_1,{\bf r}_2,{\bf r}_3,{\bf r}_4\}=\mathrm{H}[{\bf v}_1]$, which contradicts the assumption that $C$ is an emergent context. This completes the proof. 
\end{proof}

\noindent We will now establish some notable features of $(2-2)$ context.

\begin{lemma}\label{t3l5}
Let ${\bf u},{\bf v}\in\mathbb{RP}^3$ be Hamilton-distinct. If a type $(2\text{-}2)$ emergent context exists between $\mathrm{H}[{\bf u}]$ and $\mathrm{H}[{\bf v}]$, then such contexts always occur exactly in pairs.
\end{lemma}
\begin{proof}
Let $C^{\alpha\beta|\gamma\delta}_{\bf uv}:=\left\{\vec u^\alpha,\vec u^\beta,\vec v^\gamma,\vec v^\delta\right\}$ be a (2-2) type context, where $\alpha\neq\beta$ and $\gamma\neq\delta$. We will show that $C^{ab|cd}_{\bf uv}$ is also a context, where $a,b\in\{0,1,2,3\}\setminus\{\alpha,\beta\}$ with $a\neq b$ and $c, d\in\{0,1,2,3\}\setminus\{\gamma,\delta\}$ with $c\neq d$. Since, $C^{\alpha\beta|\gamma\delta}_{\bf uv}$ is a context we have 
\begin{subequations}
\begin{align}
&\mathrm{Span}\left\{\Vec u^a,\Vec u^b\right\}=\mathrm{Span}\left\{\Vec v^\gamma,\Vec v^\delta\right\},\\
&\mathrm{Span}\left\{\Vec v^c,\Vec v^d\right\}=\mathrm{Span}\left\{\Vec u^\alpha,\Vec u^\beta\right\},
\end{align}   
\end{subequations}
which proves the claim. Since ${\bf u},{\bf v}$ are Hamilton-distinct, no more type (2-2) context can be obtained .
\end{proof}

\begin{observation}\label{t3o1}
Pairs of $(2\text{-}2)$ emergent contexts among mutually Hamilton-distinct ${\bf u},{\bf v},{\bf w}\in\mathbb{RP}^3$ can occur only in certain patterns as listed below:
\begin{table}[h!]
\centering
\begin{tabular}{c|c|c}
${\bf u},{\bf v}$ & ${\bf u},{\bf w}$ & ${\bf v},{\bf w}$ \\\hline
~~~$\big\{C^{\alpha\beta|\gamma\delta}_{\bf uv},C^{ab|cd}_{\bf u \bf v}\big\}$ ~~~& ~~~$\big\{C^{\alpha\beta|\mu\nu}_{\bf uw},~C^{ab|mn}_{\bf u \bf w}\big\}$ ~~~& ~~~$\big\{C^{\gamma\delta|mn}_{\bf vw},C^{cd|\mu\nu}_{\bf vw}\big\}$~~~\\\hline
\rotatebox[origin=0]{30}{$\Big\{C^{\alpha\beta|\gamma\delta}_{\bf uv},~C^{ab|cd}_{\bf uv}\Big\}$} & \rotatebox[origin=0]{30}{$\Big\{C^{\alpha b|\mu\nu}_{\bf uw},C^{a\beta|mn}_{\bf uw}\Big\}$} & 
\begin{tabular}{c}
$\big\{C^{\gamma c|\mu m}_{\bf vw},C^{\delta d|\nu n}_{\bf vw}\big\}$;~~$\big\{C^{\gamma c|\mu n}_{\bf vw},C^{\delta d|\nu m}_{\bf vw}\big\}$;~~$\big\{C^{\gamma c|\nu m}_{\bf vw},C^{\delta d|\mu n}_{\bf vw}\big\}$~\vspace{.1cm}\\
$\big\{C^{\gamma c|\nu n}_{\bf vw},C^{\delta d|\mu m}_{\bf vw}\big\}$;~~$\big\{C^{\gamma d|\mu m}_{\bf vw},C^{\delta c|\nu n}_{\bf vw}\big\}$;~~$\big\{C^{\gamma d|\mu n}_{\bf vw},C^{\delta c|\nu m}_{\bf vw}\big\}$~\vspace{.1cm}\\
$\big\{C^{\gamma d|\nu m}_{\bf vw},C^{\delta c|\mu n}_{\bf vw}\big\}$;~~$\big\{C^{\gamma d|\nu n}_{\bf vw},C^{\delta c|\mu m}_{\bf vw}\big\}$;~~\text{No }(2-2)\text{context}~\vspace{.1cm}
\end{tabular}
\\\hline
\end{tabular}
\label{2-2-three}
\end{table}
\end{observation}

\begin{lemma}\label{t3l6}
Let ${\bf u},{\bf v},{\bf w},{\bf x}\in\mathbb{RP}^3$ be Hamilton-distinct. If a type $(1\text{-}1\text{-}1\text{-}1)$ emergent context exists among $\mathrm{H}[{\bf u}]$, $\mathrm{H}[{\bf v}]$, $\mathrm{H}[{\bf w}]$, and $\mathrm{H}[{\bf x}]$, then there exist exactly four such type $(1\text{-}1\text{-}1\text{-}1)$ emergent contexts among them.
\end{lemma}
\begin{proof}
Without loss of any generality let $C^{0000}_{\bf uvwx}:=\left\{{\bf u}^0,{\bf v}^0,{\bf w}^0,{\bf x}^0\right\}$ be a type (1-1-1-1) emergent context. Therefore we have 
\begin{align}
&\hspace{1cm}\tr\left[{\bf a}^0{\bf b}^0\right]=0,~\text{for all}~{\bf a}^0,{\bf b}^0\in\left\{{\bf u}^0,{\bf v}^0,{\bf w}^0,{\bf x}^0\right\}~\text{with}~{\bf a}^0\neq{\bf b}^0.
\end{align}
Further, Lemma~\ref{t3l2} ensures that
\begin{align}
&\tr\left[{\bf a}^\mu{\bf b}^\mu\right]=0,~\text{for all}~{\bf a}^\mu,{\bf b}^\mu\in\left\{{\bf u}^\mu,{\bf v}^\mu,{\bf w}^\mu,{\bf x}^\mu\right\}~\text{with}~{\bf a}^\mu\neq{\bf b}^\mu;~\&~\forall~\mu\in\{0,1,2,3\},
\end{align}
thereby assuring all of $C^{\mu\mu\mu\mu}_{\bf uvwx}$ to be type (1-1-1-1) emergent context for $\mu\in\{0,1,2,3\}$.

\noindent Furthermore, not more than four type (1-1-1-1) emergent context can occur in $\mathrm{H}\left[\{{\bf u},{\bf v},{\bf w},{\bf x}\}\right]$. To prove this claim, without loss of any generality we can consider:
\begin{align}
&\Vec u^0= (1,0,0,0),~~ \vec v^0 = (0,v_1,v_2,v_3),~~ \vec w^0= (0,w_1,w_2,w_3),~~\vec x^0 = (0,x_1,x_2,x_3);\nonumber\\
&\hspace{3cm}\text{s. t. } \sum_{i=1}^3a_ib_i=0;~~ a_i,b_i \in \{v_i,w_i,x_i\}~\&~ a_i\ne b_i.
\end{align}
The Hamilton extension of these vectors are

\begin{align}
\left.\begin{aligned}
\mathtt{H}[\Vec u]&\equiv\big\{\Vec u^0=(1,0,0,0),~~\quad\quad \Vec u^1=(0,1,0,0),\qquad\qquad\Vec u^2=(0,0,1,0),~~\qquad\qquad\Vec u^3=(0,0,0,1)\big\},\\
\mathtt{H}[\Vec v]&\equiv\big\{\Vec v^0=(0,v_1,v_2,v3),~\quad\Vec v^1=(-v_1,0,-v_3,v_2),~\quad\Vec v^2=(-v_2,v_3,0,-v_1),~\quad\Vec v^3=(-v_3,-v_2,v_1,0)\big\},\\
\mathtt{H}[\Vec w]&\equiv\big\{\Vec w^0=(0,w_1,w_2,w_3),~\Vec w^1=(-w_1,0,-w_3,w_2),~\Vec w^2=(-w_2,w_3,0,-w_1),~\Vec w^3=(-w_3,-w_2,w_1,0)\big\},\\
\mathtt{H}[\Vec x]&\equiv\big\{\Vec x^0=(0,x_1,x_2,x_3),\quad\Vec x^1=(-x_1,0,-x_3,x_2),\quad\Vec x^2=(-x_2,x_3,0,-x_1),~~~\Vec x^3=(-x_3,-x_2,x_1,0)\big\}
\end{aligned}\right\}.   
\end{align}

\noindent Consider the following three cases: 
\begin{itemize}
\item[\textbf{C-1:}]  Let $\{\Vec u^0,\Vec v^1,\Vec w^2,\Vec x^3\}$ forms a type (1-1-1-1) emergent context, which implies that 
\begin{align*}
v_1=w_2=x_3=0.    
\end{align*}
Furthermore, $\Vec v^0\perp\Vec w^0$ implies either $v_3=0$ or $w_3=0$. Consequently, either $\Vec v^0\in\mathrm H[\Vec u^0]$ or $\Vec w^0\in\mathrm H[\Vec u^0]$, thereby contradicting the assumption that they are Hamilton distinct.

\item[\textbf{C-2:}] Let $\{\Vec u^0,\Vec v^0,\Vec w^\mu,\Vec x^\nu\}$ is a type (1-1-1-1) emergent context, where $\mu,\nu\in\{1,2,3\}$. Then 
\begin{align*}
\Span\left\{\Vec w^0, \Vec x^0,\Vec w^\mu, \Vec x^\nu\right\}\perp\Span\{\Vec u^0,\Vec v^0\}.    
\end{align*}
Furthermore, $\Vec w^0\perp\Vec w^\mu\implies\Vec x^0=\Vec w^\mu$, again  contradicting the Hamilton distinctness assumption.

\item[\textbf{C-3:}] Let $\{\Vec u^0,\Vec v^0,\Vec w^0,\Vec x^\mu\}$ is a type (1-1-1-1) emergent context, where $\mu\in\{1,2,3\}$. This implies
\begin{align*}
\Vec x^0\perp\Span\{\Vec u^0,\Vec v^0,\Vec w^0\}\perp\Vec x^\mu.   
\end{align*}
However, this is a contradiction as $\{\Vec u^0,\Vec v^0,\Vec w^0\}$ spans a 3 dimensional space and $\Vec x^0\perp\Vec x^\mu$
\end{itemize}
\noindent Any other type (1-1-1-1) emergent context boils down to these three cases under relabeling of the indices, and hence  concludes the proof.
\end{proof}

\begin{lemma}\label{t3l7}
Let ${\bf u},{\bf v},{\bf w},{\bf x},{\bf y}\in\mathbb{RP}^3$ be Hamilton-distinct. At most one four-element subset of $\{{\bf u},{\bf v},{\bf w},{\bf x},{\bf y}\}$ can give rise to a type $(1\text{-}1\text{-}1\text{-}1)$ emergent context.
\end{lemma}
\begin{proof}
Without loss of generality let type (1-1-1-1) emergent context appears among $\bf u, \bf v, \bf w, \bf x$, and suppose they are of the form  $C_{\bf uvwx}^{\mu\mu\mu\mu}$ for $\mu\in\{0,1,2,3\}$. Contrary to the claim, without loss of generality, let us assume that $\bf u, \bf v, \bf w, \bf y$ also give rise to type (1-1-1-1) emergent contexts. Analyzing the following cases suffice our purpose. 

\begin{itemize}
\item[\textbf{C-1:}] Emergent context of form $C_{\bf uvwy}^{\mu\mu\mu\xi}$: The context $C_{\bf uvwx}^{\mu\mu\mu\mu}$ along with the fact that the projectors are Hamilton distinct restricts $\bf y^\xi$ to be orthogonal to at most two among $\bf u^\mu, \bf v^\mu, \bf w^\mu$, thereby prohibiting $C_{\bf uvwy}^{\mu\mu\mu\xi}$ to be a type (1-1-1-1) emergent context. 

\item[\textbf{C-2:}] Emergent context of form $C_{\bf uvwy}^{\mu\mu\nu\xi}$: Without loss of generality, the vectors ${\bf u},{\bf v},{\bf w},{\bf x},{\bf y}\in\mathbb{RP}^3$ and their and Hamilton extension can be taken as of the following form:
\begin{align}
\left.\begin{aligned}
\mathtt{H}[\Vec u]&\equiv\big\{\Vec u^0=(1,0,0,0),~~~\qquad \Vec u^1=(0,1,0,0),~~~\quad\qquad\Vec u^2=(0,0,1,0),~~~~\quad\qquad\Vec u^3=(0,0,0,1)\big\},\\
 \mathtt{H}[\Vec v]&\equiv\big\{\Vec v^0=(0,v_1,v_2,v_3),~~\quad\Vec v^1=(-v_1,0,-v_3,v_2),\quad\Vec v^2=(-v_2,v_3,0,-v_1),~\quad\Vec v^3=(-v_3,-v_2,v_1,0)\big\},\\ 
\mathtt{H}[\Vec w]&\equiv\big\{\Vec w^0=(0,w_1,w_2,w_3),~\Vec w^1=(-w_1,0,-w_3,w_2),~\Vec w^2=(-w_2,w_3,0,-w_1),~\Vec w^3=(-w_3,-w_2,w_1,0)\big\},\\ 
\mathtt{H}[\Vec x]&\equiv\big\{\Vec x^0=(0,x_1,x_2,x_3),~~~~\Vec x^1=(-x_1,0,-x_3,x_2),~~~\Vec x^2=(-x_2,x_3,0,-x_1),~~~\Vec x^3=(-x_3,-x_2,x_1,0)\big\},\\ 
\mathtt{H}[\Vec y]&\equiv\big\{\Vec y^0=(y_0,y_1,x_2,x_3),~~~\Vec y^1=(-x_1,y_0,-x_3,x_2),~\Vec y^2=(-y_2,y_3,y_0,-y_1),~~~\Vec y^3=(-y_3,-y_2,y_1,y_0)\big\}
\end{aligned}\right\}.   
\end{align}

\noindent Furthermore, analyzing the case $C_{\bf uvwy}^{00\nu\xi} $ is sufficient. Existence of an emergent context $C_{\bf uvwy}^{00\nu\xi}$ implies $w_\nu=0$ and $w_{\nu\oplus_3 1} = v_{\nu\oplus_3 1} , w_{\nu\oplus_3 2} = v_{\nu\oplus_3 2}$. But, $\Vec v^0 \perp \Vec w^0\implies v_{\nu\oplus_3 1} = v_{\nu\oplus_3 2}=0$, thereby forcing $\Vec w $ to be the zero vector, and thus leading to a contradiction.

\item[\textbf{C-3:}] Emergent context of form $C_{\bf uvwy}^{\mu\nu\eta\xi}$: Analyzing the case $C_{\bf uvwy}^{012\xi} $ is sufficient. The context $C_{\bf uvwy}^{012\xi}$ implies $v_1 = w_2 =0$ and $\Vec v^0 \perp \Vec w^0\implies v_3 = 0 \text{ or }w_3=0$. Which further enforces either $\Vec v^0\in \mathsf{H}[\Vec u]$ or $\Vec w^0\in \mathsf{H}[\Vec u]$, thereby leading to a contradiction.
\end{itemize}
\noindent This completes the proof. 
\end{proof}

\begin{lemma}\label{t3l8}
Let ${\bf u},{\bf v},{\bf w},{\bf x},{\bf y}\in\mathbb{RP}^{3}$ be pairwise Hamilton-distinct vectors. If $\{{\bf u},{\bf v},{\bf w},{\bf x}\}$ induces a type (1-1-1-1) emergent context, then among the four pairs $({\bf y},{\bf u})$, $({\bf y},{\bf v})$, $({\bf y},{\bf w})$, and $({\bf y},{\bf x})$, at most two can generate a type (2-2) emergent context.
\end{lemma}

\begin{proof}
Without loss of generality, assume that ${\bf y}$ forms type (2-2) emergent contexts with both ${\bf u}$ and ${\bf v}$. Then, up to a Hamilton-preserving change of basis, the corresponding vectors may be chosen as follows

\begin{align}
\left.\begin{aligned}
\mathtt{H}[\Vec u]&\equiv\big\{\Vec u^0=(1,0,0,0),~~~~\qquad \Vec u^1=(0,1,0,0),~\qquad\qquad\Vec u^2=(0,0,1,0),~~~\qquad\Vec u^3=(0,0,0,1)\big\},\\
\mathtt{H}[\Vec v]&\equiv\big\{\Vec v^0=(0,0,v_2,v_3),~\qquad\Vec v^1=(0,0,-v_3,v_2),~~\qquad\Vec v^2=(-v_2,v_3,0,0),~\quad\Vec v^3=(-v_3,-v_2,0,0)\big\},\\
\mathtt{H}[\Vec w]&\equiv\big\{\Vec w^0=(0,w_1,-v_3,v_2),~\Vec w^1=(-w_1,0,-v_2,-v_3),~\Vec w^2=(v_3,v_2,0,-w_1),~\Vec w^3=(-v_2,v_3,w_1,0)\big\},\\
\mathtt{H}[\Vec x]&\equiv\big\{\Vec x^0=(0,x_1,-v_3,v_2),~~~\Vec x^1=(-x_1,0,-v_2,-v_3),~~\Vec x^2=(v_3,v_2,0,-x_1),~~\Vec x^3=(-v_2,v_3,x_1,0)\big\},\\
\mathtt{H}[\Vec y]&\equiv\big\{\Vec y^0=(0,0,y_2,y_3),~\qquad\Vec y^1=(0,0,-y_3,y_2),~~\qquad\Vec y^2=(-y_2,y_3,0,0),\quad\Vec y^3=(y_3,y_2,0,0)\big\}
\end{aligned}\right\}.   
\end{align}
\noindent Consider any vector $\Vec w^\mu$ from the Hamilton context $\mathtt{H}[\Vec w^0]=\{\Vec w^0,\Vec w^1,\Vec w^2,\Vec w^3\}$. Suppose that $\Vec y^0$ is orthogonal to some $\Vec w^\mu$. Using the explicit forms of the vectors, the orthogonality condition $\langle \Vec y^0,\Vec w^\mu\rangle =0$ implies either
\begin{itemize}
\item[(i)] $\Vec y^0$ belongs to the Hamilton context of $\Vec v^0$, or
\item[(ii)] one of the vectors of $\mathtt{H}[\Vec w^0]$ belongs to the Hamilton context of $\Vec v^0$, 
\end{itemize}
thereby contradicting the assumption of pairwise Hamilton-distinctness. An identical argument applies to $\mathtt{H}[\Vec x^0]$. Hence ${\bf y}$ cannot form a type $(2,2)$ emergent context with either ${\bf w}$ or ${\bf x}$. This completes the proof.
\end{proof}

\begin{lemma}\label{t3l9}
Let ${\bf u},{\bf v},{\bf w},{\bf x}\in\mathbb{RP}^{3}$ be pairwise Hamilton-distinct and form type (1-1-1-1) emergent contexts. Suppose that a type (2-2) emergent context exists between two vectors ${\bf a}$ and ${\bf b}$ from the set $\{{\bf u},{\bf v},{\bf w},{\bf x}\}$. Let, $\{{\bf c},{\bf d}\}=\{{\bf u},{\bf v},{\bf w},{\bf x}\}\setminus\{{\bf a},{\bf b}\}$. Then (1) ${\bf c}$ and ${\bf d}$ also form a type (2-2) emergent context; and (2) no type $(2,2)$ emergent context exists between any of the pairs $({\bf a},{\bf c}),~({\bf a},{\bf d}),~({\bf b},{\bf c}),~({\bf b},{\bf d})$.
\end{lemma}

\begin{proof}
Without loss of generality, assume that ${\bf u}$ and ${\bf v}$ form a type (2-2) emergent context. Then the corresponding vectors may be chosen of the form of Lemma~\ref{t3l8}. From the explicit forms of the vectors, we observe that $\langle \Vec w^0,\Vec x^0\rangle=\langle \Vec w^1,\Vec x^1\rangle=0$. Hence, $\left\{\Vec w^0,\Vec w^1,\Vec x^0,\Vec x^1\right\}$ forms a type (2-2) emergent context.

\noindent It remains to show that no type (2-2) emergent context can exist between a vector from $\{{\bf u},{\bf v}\}$ and a vector from $\{{\bf w},{\bf x}\}$. By symmetry, it suffices to consider the pair $({\bf u},{\bf w})$. Suppose that ${\bf u}$ and ${\bf w}$ form a type $(2,2)$ emergent context. Then there exists $\mu\in{1,2,3}$ such that $\langle \Vec w^0,\Vec u^\mu\rangle=0$. If $\mu=1$, then $w_1=0$, and consequently $\Vec w^1=(0,0,-v_2,-v_3)\in\mathtt{H}[\Vec v^0]$. Thus ${\bf w}$ and ${\bf v}$ are not Hamilton-distinct, a contradiction. If $\mu=2$ or $\mu=3$, then $v_2=0$ or $v_3=0$, respectively. In either case one obtains $\Vec v^0\in\mathtt{H}[\Vec u^0]$, contradicting the assumption that ${\bf u}$ and ${\bf v}$ are Hamilton-distinct. By the same argument, none of the pairs $({\bf u},{\bf x}),~({\bf v},{\bf w}),~({\bf v},{\bf x})$ can form a type $(2,2)$ emergent context.
\end{proof}

\begin{lemma}\label{t3l10}
Let ${\bf u},{\bf v},{\bf w},{\bf x},{\bf y}\in\mathbb{RP}^{3}$
be pairwise Hamilton-distinct. Assume that $\{{\bf u},{\bf v},{\bf w},{\bf x}\}$ forms a type (1-1-1-1) emergent context. If ${\bf a}$ forms type (2-2) emergent contexts with both ${\bf b}$ and ${\bf y}$, where ${\bf a},{\bf b}\in\{{\bf u},{\bf v},{\bf w},{\bf x}\}$, then ${\bf y}$ cannot form a type (2,2) emergent context with any other vector of $\{{\bf u},{\bf v},{\bf w},{\bf x}\}\setminus\{{\bf a}, {\bf b}\}$.
\end{lemma}

\begin{proof}
Without loss of generality, suppose that there is a type (2-2) emergent context between ${\bf u}$ and ${\bf v}$ and another between ${\bf u}$ and ${\bf y}$. Let these contexts be are $C_{\bf uv}^{0\mu|0\nu}$ and $C_{\bf uy}^{0\xi|0\eta}$. We distinguish two cases.

\smallskip
\noindent
\textbf{Case 1: $(\mu=\xi)$.} In this case, Observation~\ref{t3o1} implies that there exists a type (2-2) emergent context between ${\bf v}$ and ${\bf y}$. Consequently, ${\bf y}$ forms type $(2,2)$ emergent contexts with both ${\bf u}$ and ${\bf v}$. Since $\{{\bf u},{\bf v},{\bf w},{\bf x}\}$ forms a type (1-1-1-1) emergent context, Lemma~\ref{t3l8} implies that ${\bf y}$ cannot form a type (2-2) emergent context with either ${\bf w}$ or ${\bf x}$.

\smallskip
\noindent
\textbf{Case 2: $(\mu\neq\xi)$.} Without loss of generality, let the contexts are $C_{{\bf uv}}^{01|01}$, $C_{{\bf uy}}^{02|02}$ and the corresponding vectors are of the form
\begin{align}
\left.\begin{aligned}
\mathtt{H}[\Vec u]&\equiv\big\{\Vec u^0=(1,0,0,0),~~~\quad\quad \Vec u^1=(0,1,0,0),~\qquad\qquad\Vec u^2=(0,0,1,0),~~~\quad\qquad\Vec u^3=(0,0,0,1)\big\},\\
\mathtt{H}[\Vec v]&\equiv\big\{\Vec v^0=(0,0,v_2,v_3),~~~~\quad\Vec v^1=(0,0,-v_3,v_2),~~~\quad\quad\Vec v^2=(-v_2,v_3,0,0),~~~~\quad\Vec v^3=(-v_3,-v_2,0,0)\big\},\\
\mathtt{H}[\Vec w]&\equiv\big\{\Vec w^0=(0,w_1,-v_3,v_2),~\Vec w^1=(-w_1,0,-v_2,-v_3),~\Vec w^2=(v_3,v_2,0,-w_1),~~~~~~\Vec w^3=(-v_2,v_3,w_1,0)\big\},\\
\mathtt{H}[\Vec x]&\equiv\big\{\Vec x^0=(0,x_1,v_3,-v_2),~~~\Vec x^1=(-x_1,0,v_2,v_3),~~~~\quad\Vec x^2=(-v_3,-v_2,0,-x_1),~\Vec x^3=(v_2,-v_3,x_1,0)\big\},\\
\mathtt{H}[\Vec y]&\equiv\big\{\Vec y^0=(0,y_1,0,y_3),~~~~\quad\Vec y^1=(-y_1,0,-y_3,0),~~~~\quad\Vec y^2=(0,y_3,0,-y_1),~\qquad\Vec y^3=(-y_3,0,y_1,0)\big\}
\end{aligned}\right\}.   
\end{align}
Suppose, towards a contradiction, that ${\bf y}$ forms a type (2-2) emergent context with ${\bf w}$. Then $\Vec y^0$ must be orthogonal to two vectors in $\mathtt{H}[\Vec w]$. A direct examination of the orthogonality equations $\langle \Vec y^0,\Vec w^\alpha\rangle =\langle \Vec y^0,\Vec w^\beta\rangle =0$ for distinct $\alpha,\beta\in{0,1,2,3}$, shows that either ${\bf y}\in\mathtt{H}[{\bf u}]$ or ${\bf v}\in\mathtt{H}[{\bf u}]$. Both possibilities contradict the assumption of pairwise Hamilton-distinctness. Hence no type (2-2) emergent context can exist between ${\bf y}$ and ${\bf w}$. By symmetry, the same conclusion holds for ${\bf y}$ and ${\bf x}$. This concludes the proof.
\end{proof}

\subsection{A-4: Proof of Proposition~\ref{prop1}}
\begin{proof}
We first establish the claim for a singleton set $\mathcal V=\{\vec v\}$. Using the quaternion representation above,
\begin{align}
\mathrm H[\mathrm H[\vec v]]\equiv\left\{\mathrm e_\alpha\mathrm e_\beta\theta(\vec v)\,:\,\alpha,\beta\in\{0,1,2,3\}\right\}.    
\end{align}
Since the quaternion basis is closed under multiplication up to sign, for every pair $\alpha,\beta\in\{0,1,2,3\}$ there exists $\gamma\in\{0,1,2,3\}$ such that $\mathrm e_\alpha\mathrm e_\beta=\pm \mathrm e_\gamma$. Hence,
\begin{align}
\mathrm H[\mathrm H[\vec v]]\equiv\left\{\pm \mathrm e_\gamma\theta(\vec v)\,:\,\gamma=0,1,2,3\right\}.   \end{align}
Throughout this work, vectors differing by an overall sign are considered identical, namely, $\vec v\sim -\vec v$. Consequently it follows that 
\begin{align}
\mathrm H[\mathrm H[\vec v]]=\mathrm H[\vec v].\label{single}    
\end{align}

\noindent Now let
$\mathcal V=\{\vec v_i\}_{i=1}^{n}\subseteq\mathbb R^{4}$. We first argue that $\mathrm H[\mathrm H[\mathcal V]\subseteq\mathrm H[\mathcal V]$. Let $w\in \mathrm H[\mathrm H[\mathcal V]]$. Then, by definition of Hamilton extension of a set, there exists some vector $u\in \mathrm H[\mathcal V]$ such that $w\in \mathrm H[u]$. Again, since $u\in \mathrm H[\mathcal V]:=\bigcup_{i=1}^{n}\mathrm H[\vec v_i]$, there exists some index $j\in\{1,\dots,n\}$ such that $u\in \mathrm H[\vec v_j]$. Therefore, $w\in \mathrm H[\mathrm H[\vec v_j]]$. Using Eq.~(\ref{single}) we have $w\in \mathrm H[\vec v_j]\subseteq\mathrm H[\mathcal V]$, which establishes $\mathrm H[\mathrm H[\mathcal V]]\subseteq\mathrm H[\mathcal V]$. On the other hand, $\vec v\in \mathrm H[\vec v],~\forall~\vec v\in\mathbb R^4$, thereby implying $\mathrm H[\mathcal V]\subseteq\mathrm H[\mathrm H[\mathcal V]]$. We therefore have 
\begin{align}
\mathrm H[\mathrm H[\mathcal V]]=\mathrm H[\mathcal V],~\text{for an arbitrary  set}~ \mathcal V\subseteq\mathbb{R}^4.   
\end{align}
This concludes the proof.
\end{proof}

\section{Appendix B: Theorem~\ref{theo1} and Theorem~\ref{theo2}}
\section{B-1. Proof of Theorem~\ref{theo1}}
\begin{proof}
Hamilton extension $\mathrm{H}\left[\mathcal{V}_8\right]$ of the set $\mathcal{V}_8\equiv\left\{\Vec{u}_i\right\}_{i=1}^8$ contains $32$ distinct vectors as listed in Table~\ref{tabs2}. One can construct $18$ measurement contexts from this set:
\begin{subequations}
\begin{align}
&~C_i\equiv\mathrm{H}\left[\Vec{u}_i\right]=\left\{{\bf u}^0_i,{\bf u}^1_i,{\bf u}^2_i,{\bf u}^3_i\right\},~~\text{for}~i\in\{1,2,\ldots,8\}~~\Big[\text{Hamilton context}\Big],\\
&\left.\begin{aligned}
&C^{01|01}_{\bf 12}=\left\{{\bf u}^0_1,{\bf u}^1_1,{\bf u}^0_2,{\bf u}^1_2\right\},~~C^{23|23}_{\bf 12}=\left\{{\bf u}^2_1,{\bf u}^3_1,{\bf u}^2_2,{\bf u}^3_2\right\};\\
&C^{02|02}_{\bf 13}=\left\{{\bf u}^0_1,{\bf u}^2_1,{\bf u}^0_3,{\bf u}^2_3\right\},~~C^{13|13}_{\bf 13}=\left\{{\bf u}^1_1,{\bf u}^3_1,{\bf u}^1_3,{\bf u}^3_3\right\};\\
&C^{03|03}_{\bf 14}=\left\{{\bf u}^0_1,{\bf u}^3_1,{\bf u}^0_4,{\bf u}^3_4\right\},~~C^{12|12}_{\bf 14}=\left\{{\bf u}^1_1,{\bf u}^2_1,{\bf u}^1_4,{\bf u}^2_4\right\}
\end{aligned}\right\}~\Big[\text{Type (2-2) context}\Big],\\
&\left.\begin{aligned}
&C^{1|0|2|3}_{\bf 5678}=\left\{{\bf u}^1_5,{\bf u}^0_6,{\bf u}^2_7,{\bf u}^3_8\right\},~~
C^{0|1|3|2}_{\bf 5678}=\left\{{\bf u}^0_5,{\bf u}^1_6,{\bf u}^3_7,{\bf u}^2_8\right\},\\
&C^{2|3|1|0}_{\bf 5678}=\left\{{\bf u}^2_5,{\bf u}^3_6,{\bf u}^1_7,{\bf u}^0_8\right\},~~
C^{3|2|0|1}_{\bf 5678}=\left\{{\bf u}^3_5,{\bf u}^2_6,{\bf u}^0_7,{\bf u}^1_8\right\}
\end{aligned}\right\}~\Big[\text{Type (1-1-1-1) context}\Big].
\end{align}
\end{subequations}

\begin{table}[h!]
\centering
\begin{tabular}{|c|c||c|c|c|c|}
\hline
$\#$ & $\Vec{u}_i$ & $\Vec{u}^0_i$ & $\Vec{u}^1_i$ & $\Vec{u}^2_i$ & $\Vec{u}^3_i$ \\ \hline
1 &~~~ $(0,0,0,1)$~~~ &~~~$(0,0,0,1)$~~~  &~~~ $(0,0,1,0)$ ~~~&~~~$(0,1,0,0)$ ~~~ &~~~$(1,0,0,0)$~~~  \\ \hline
2 & $(1,-1,0,0)$ &$(1,-1,0,0)$  &(1,1,0,0)  & $(0,0,1,1)$ &$(0,0,1,-1)$  \\ \hline
3 & $(1,0,-1,0)$ & $(1,0,-1,0)$ &$(0,1,0,-1)$  & $(1,0,1,0)$ &$(0,1,0,1)$  \\ \hline
4 & $(0,1,-1,0)$ & $(0,1,-1,0)$&$(1,0,0,1)$  &$(1,0,0,-1)$  &$(0,1,1,0)$    \\ \hline
5 & $(1,-1,1,0)$ &$(1,-1,1,0)$  &$(1,1,0,1)$  &$(-1,0,1,1)$  &$(0,1,1,-1)$  \\ \hline
6 & $(-1,1,1,0)$ & $(-1,1,1,0)$ & $(1,1,0,-1)$ & $(1,0,1,1)$ & $(0,1,-1,1)$ \\ \hline
7 & $(1,1,-1,0)$ &$(1,1,-1,0)$  & $(1,-1,0,1)$ & $(1,0,1,-1)$ &  $(0,1,1,1)$\\ \hline
8 & $(1,1,1,0)$ &$(1,1,1,0)$  &$(-1,1,0,1)$  &$(1,0,-1,1)$  &$(0,-1,1,1)$  \\ \hline
\end{tabular}
\caption{The set of vectors $\mathcal{V}_8=\left\{\Vec{u}_i\right\}_{i=1}^8\subset\mathbb{R}^4$ appearing in Theorem~\ref{theo1}, and their Hamilton extension $\mathrm{H}\left[\mathcal{V}_8\right]$.}
\label{tabs2}
\end{table}
\noindent For any KS coloring $\mu:\mathrm{H}\left[\mathcal{V}_{8}\right]\to\{0,1\}$, the completeness conditions demands exactly one vector in each context to be assigned value $1$, whereas the exclusivity condition demands no two orthogonal vectors to be assigned value $1$. It is to note that the construction $\mathrm{H}\left[\mathcal{V}_{8}\right]$ is symmetric with respect to the projectors $\{{\bf u}_1^0,{\bf u}_1^1,{\bf u}_1^2,{\bf u}_1^3\}$. Without loss of any generality, let assign $\mu({\bf u}_1^3)=1$ in context $C_1$. Consequently, we have 
\begin{align}
\mu({\bf u}_1^0)=\mu({\bf u}_1^1)=\mu({\bf u}_1^2)=0.   
\end{align}
Completeness requirement then demands exactly one of the projectors $\{{\bf u}_4^1,{\bf u}_4^2\}$ to be assigned value $1$ in context $C^{12|12}_{\bf 14}$. Similarly, exactly one of $\{{\bf u}_3^0,{\bf u}_3^2\}$ in context $C^{02|02}_{\bf 13}$ and exactly one of the vectors $\{{\bf u}_2^0,{\bf u}_2^1\}$ in context $C^{01|01}_{\bf 12}$ to be assigned value $1$, thereby leading to eight different possible assignments:
\begin{align}
\left.\begin{aligned}
\mu({\bf u}^1_4)=\mu({\bf u}^0_3)=\mu({\bf u}^0_2)=1;\qquad\mu({\bf u}^1_4)=\mu({\bf u}^0_3)=\mu({\bf u}^1_2)=1;\\
\mu({\bf u}^1_4)=\mu({\bf u}^2_3)=\mu({\bf u}^0_2)=1;\qquad\mu({\bf u}^1_4)=\mu({\bf u}^2_3)=\mu({\bf u}^1_2)=1;\\
\mu({\bf u}^2_4)=\mu({\bf u}^0_3)=\mu({\bf u}^0_2)=1;\qquad\mu({\bf u}^2_4)=\mu({\bf u}^0_3)=\mu({\bf u}^1_2)=1;\\
\mu({\bf u}^2_4)=\mu({\bf u}^2_3)=\mu({\bf u}^0_2)=1;\qquad\mu({\bf u}^2_4)=\mu({\bf u}^2_3)=\mu({\bf u}^1_2)=1;
\end{aligned}\right\}.
\end{align}
Consider the assignment $\mu({\bf u}_4^1)=\mu({\bf u}_3^0)=\mu({\bf u}_2^0)=\mu({\bf u}_1^3)=1$. As can be checked from Table~\ref{tabs2}, we have
\begin{align}
{\bf u}_4^1\perp {\bf u}_8^1,\quad {\bf u}_3^0\perp {\bf u}_6^2,\quad {\bf u}_2^0\perp {\bf u}_7^0,\quad {\bf u}_1^3\perp {\bf u}_5^3.
\end{align}
The exclusivity condition thus demands 
\begin{align}
\mu({\bf u}_8^1)=\mu({\bf u}_6^2)=\mu({\bf u}_7^0)=\mu({\bf u}_5^3)=0.
\end{align}
However, this violates the completeness condition when considering the emergent context $C^{3|2|0|1}_{\bf 5678}$. Similar contradiction arise for other possible assignments:
\begin{align}
\left.\begin{aligned}
&\Big\{\mu({\bf u}^1_4)=\mu({\bf u}^0_3)=\mu({\bf u}^0_2)=\mu({\bf u}_1^3)=1\Big\}\mathrlap{\longrightarrow}\times~~~C^{3|2|0|1}_{\bf 5678};\quad \Big\{\mu({\bf u}^1_4)=\mu({\bf u}^0_3)=\mu({\bf u}^1_2)=\mu({\bf u}_1^3)=1\Big\}\mathrlap{\longrightarrow}\times~~~ C_{\bf 6};\\
&\Big\{\mu({\bf u}^1_4)=\mu({\bf u}^2_3)=\mu({\bf u}^0_2)=\mu({\bf u}_1^3)=1\Big\}\mathrlap{\longrightarrow}\times~~~ C^{1|0|2|3}_{\bf 5678} ;\quad \Big\{\mu({\bf u}^1_4)=\mu({\bf u}^2_3)=\mu({\bf u}^1_2)=\mu({\bf u}_1^3)=1\Big\} \mathrlap{\longrightarrow}\times~~~ C^{0|1|3|2}_{\bf 5678} ;\\
&\Big\{\mu({\bf u}^2_4)=\mu({\bf u}^0_3)=\mu({\bf u}^0_2)=\mu({\bf u}_1^3)=1\Big\}\mathrlap{\longrightarrow}\times~~~ C_{\bf 8};~~~~\qquad \Big\{\mu({\bf u}^2_4)=\mu({\bf u}^0_3)=\mu({\bf u}^1_2)=\mu({\bf u}_1^3)=1\Big\}\mathrlap{\longrightarrow}\times~~~ C_{\bf 8};\\
&\Big\{\mu({\bf u}^2_4)=\mu({\bf u}^2_3)=\mu({\bf u}^0_2)=\mu({\bf u}_1^3)=1\Big\}\mathrlap{\longrightarrow}\times~~~ C^{2|3|1|0}_{\bf 5678};\quad\Big\{\mu({\bf u}^2_4)=\mu({\bf u}^2_3)=\mu({\bf u}^1_2)=\mu({\bf u}_1^3)=1\Big\}\mathrlap{\longrightarrow}\times~~~ C_{\bf 5};
\end{aligned}\right\},
\end{align}
thereby excluding all possible valid KS colorings. This completes the proof.
\end{proof}

\begin{figure*}[t!]
\centering
\begin{tabular}{cc}
\includegraphics[height=6cm]{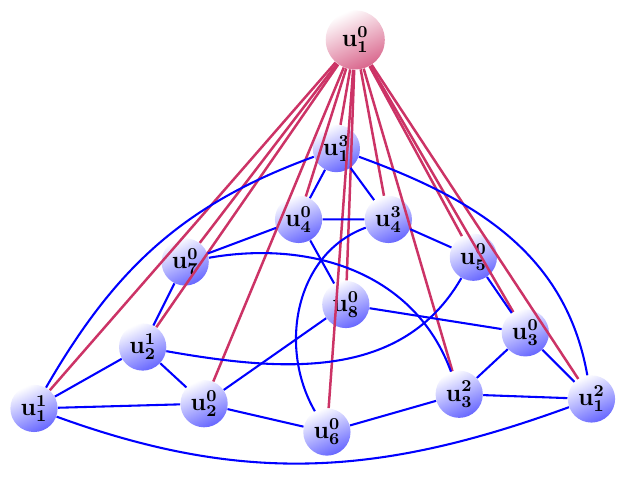}
~~
\includegraphics[height=6cm]{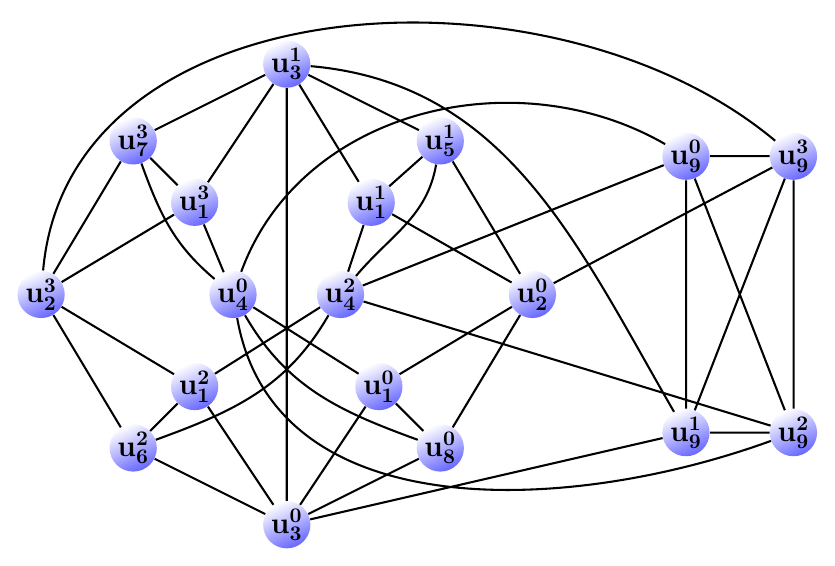}
\end{tabular}
\caption{(Color online) Left: Orthogonality graph of the set $\mathtt{YO_{14}}$. The red vertex (top) is the apex vertex. The graph obtained by removing the apex admits an orthogonal representation in $\mathbb{R}^{3}$, while the full graph admits one only in $\mathbb{R}^{4}$. Right: Orthogonality graph of the set $\mathtt{CN_{18}}$. For both graphs, the classical chromatic number exceeds the corresponding quantum chromatic number \cite{Cameron2007,Maninska2016}.} \vspace{-.25cm}
\label{figs3}
\end{figure*}

\begin{remark}
At this point, we note that an analogous argument applies when the parent set is chosen as the seven-vector set $\mathcal{V}_8\setminus\{\vec u\}$ with $\vec u\in\{\vec v_2,\vec v_3,\vec v_4\}$ (pointed out in Construction~\ref{const2} of the main manuscript). Specifically, one can show that the corresponding Hamilton extension $\mathrm{H}[\mathcal{V}_8\setminus\{\vec u\}]$ is KS-uncolorable. Consequently, the three parent sets $\mathcal{V}_8\setminus\{\vec u\}$, with $\vec u\in\{\vec v_2,\vec v_3,\vec v_4\}$, provide examples of irreducible seven-vector parent sets whose Hamilton extensions exhibit logical KS contextuality.     
\end{remark}

\subsection{B-2. Vector sets $\mathtt{YO_{14}}$ and $\mathtt{CN_{18}}$}

\noindent{\bf Vector sets $\mathtt{YO_{14}}$:} For a vector $\vec{v}=(a,b,c,d)\in\mathbb{R}^4$ let $\tilde{v}$ denotes a vector in $\mathbb{R}^3$ obtained by ignoring its fourth component, i.e. $\tilde{v}=(a,b,c)\in\mathbb{R}^3$. Then the $13$ vectors of Yu-Oh's construction \cite{Yu2012} are given by
\begin{align}
\mathtt{YO_{13}}\equiv
\Big\{\tilde{u}_1^1,~\tilde{u}_1^2,~\tilde{u}_1^3,~\tilde{u}_2^0,~\tilde{u}_2^1,~\tilde{u}_3^0,~\tilde{u}_3^2,~\tilde{u}_4^0,~\tilde{u}_4^3,~\tilde{u}_5^0,~\tilde{u}_6^0,~\tilde{u}_7^0,~\tilde{u}_8^0\Big\}\subset\mathbb{R}^3,
\end{align}
where $\vec{u}^\alpha_p$'s are the vectors appearing in $\mathrm{H}\left[\mathcal{V}_8\right]$. The $14$ vectors set $\mathtt{YO_{14}}$  \cite{Cabello2015,Maninska2016} is obtained by embedding the vectors of $\mathtt{YO_{13}}$ in $\mathbb{R}^4$ and adding an apex vertex. Accordingly the set is given by 
\begin{align}
\mathtt{YO_{14}}\equiv\Big\{\vec u_1^0,~\vec u_1^1,~\vec u_1^2,~\vec u_1^3,~\vec u_2^0,~\vec u_2^1,~\vec u_3^0,~\vec u_3^2,~\vec u_4^0,~\vec u_4^3,~\vec u_5^0,~\vec u_6^0,~\vec u_7^0,~\vec u_8^0\Big\}\subset\mathbb{R}^4.   
\end{align}
Since $\mathcal{V}_{8}\subset\mathtt{YO_{14}}$, and no vector in $\mathtt{YO_{14}}\setminus\mathcal{V}_{8}$ is Hamilton-distinct to all the vectors in $\mathcal{V}_{8}$, we thus have $\mathrm{H}[\mathtt{YO_{14}}]=\mathrm{H}[\mathcal{V}_{8}]$. As shown in \cite{Yu2012} the set $\mathtt{YO_{13}}$ exhibits SI-C. But in the set $\mathtt{YO_{14}}$ contextuality weakens to SD-C \cite{Cabello2015}. Nonetheless the Hamilton extended set $\mathrm{H}\left[\mathtt{YO_{14}}\right]=\mathrm{H}\left[\mathtt{YO_{13}}+1\right]$ exhibits logical KS contextuality. \\

\noindent{\bf Vector sets $\mathtt{CN_{18}}$:} The $18$-vector set $\mathtt{CN}_{18}$, appearing in the graph coloring game of Cameron and Newman \textit{et al.}~\cite{Cameron2007}, contains $14$ vectors from $\mathrm{H}\left[\mathcal{V}_8\right]$ along with $4$ additional vectors. The set is given by
\begin{align}
\mathtt{CN_{18}}\equiv\Big\{~\vec u_2^3,~\vec u_1^3,~\vec u_7^3,~\vec u_3^1,~\vec u_1^1,~\vec u_5^1,~\vec u_2^0,~\vec u_1^0,~\vec u_8^0,~\vec u_3^0,~\vec u_1^2,~\vec u_6^2,~\vec u_4^0,u_4^2,~\vec u_9^0,~\vec u_9^1,~\vec u_9^2,~\vec u_9^3\Big\}\subset\mathbb{R}^4,     
\end{align}
where $\vec u_9=(1,1,1,1)$ and $\mathrm{H}\left[\vec u_9\right]=\left\{\vec u_9^0,~\vec u_9^1,~\vec u_9^2,~\vec u_9^3\right\}$. Accordingly, we have 
\begin{align}
\mathrm{H}\left[\mathtt{CN_{18}}\right]=\mathrm{H}\left[\mathcal{V}^\star_{9}\right]=\mathrm{H}\left[\mathcal{V}_8\right]\cup\left\{\vec u_9^0,~\vec u_9^1,~\vec u_9^2,~\vec u_9^3\right\}, ~\text{where}~ \mathcal{V}^\star_{9}= \mathcal{V}_8\cup\left\{\vec u_9\right\}. 
\end{align}
Manifestly the set $\mathcal{V}^\star_{9}$ is not irreducible as Hamilton extension of its proper subset $\mathcal{V}_8$ exhibits logical KS contextuality.  At this point we would like to point out that the set $\mathtt{YO_{14}}$ is not contained in $\mathtt{CN_{18}}$, and they are structurally different. Nonetheless, both these sets containing $\mathcal{V}_8$ as a subset within them exhibit logical KS contextuality under Hamilton extension due to the ‘hereditary under supersets’ feature. For the sake of completeness we depict the orthogonality graphs of $\mathtt{YO_{14}}$ and $\mathtt{CN_{18}}$ in Fig.~\ref{figs3}.

\subsection{B-3. Four Cabello sets form  $\mathrm{H}\left[\mathcal{V}_{6}\right]$}
\noindent As pointed out in Fig.~\ref{fig2}, four distinct Cabello sets can be obtained from the vectors of $\mathrm{H}\left[\mathcal{V}_{6}\right]$ by removing different suitable subsets of six vectors. The resulting sets, together with their corresponding contexts, are listed in Table~\ref{tab2}.

\begin{table}[h!]
\begin{tabular}{|c|c|}
\hline
Cabello set: $\mathtt{C^{[i]}_{18}}$   &  Contexts \\ \hline
$\mathtt{H}\left[\mathcal{V}_6\right]\setminus\left\{\mathbf{u}_1^2,\mathbf{u}_2^1,\mathbf{u}_3^2,\mathbf{u}_4^1,
\mathbf{u}_5^2,\mathbf{u}_6^1\right\}$ 
&~$C^{03|03}_{\bf 12},C^{02|13}_{\bf 23},C^{03|03}_{\bf 34},C^{02|13}_{\bf 45},C^{03|03}_{\bf 56},C^{02|13}_{\bf 61},C^{01|23}_{\bf 14},C^{23|01}_{\bf 25},C^{01|23}_{\bf 36}$ \vspace{.1cm}\\\hline
$\mathtt{H}\left[\mathcal{V}_6\right]\setminus\left\{\mathbf{u}_1^1,\mathbf{u}_2^2,\mathbf{u}_3^1,\mathbf{u}_4^2,
\mathbf{u}_5^1,\mathbf{u}_6^2\right\}$ &~$C^{03|03}_{\bf 12},C^{13|02}_{\bf 23},C^{03|03}_{\bf 34},C^{13|02}_{\bf 45},C^{03|03}_{\bf 56},C^{13|02}_{\bf 61},C^{23|01}_{\bf 14},C^{01|23}_{\bf 25},C^{23|01}_{\bf 36}$ \vspace{.1cm}\\\hline
$\mathtt{H}\left[\mathcal{V}_6\right]\setminus\left\{\mathbf{u}_1^0,\mathbf{u}_2^3,\mathbf{u}_3^0,\mathbf{u}_4^3,
\mathbf{u}_5^0,\mathbf{u}_6^3\right\}$ &~$C^{12|12}_{\bf 12},C^{02|13}_{\bf 23},C^{12|12}_{\bf 34},C^{02|13}_{\bf 45},C^{12|12}_{\bf 56},C^{02|13}_{\bf 61},C^{23|01}_{\bf 14},C^{01|23}_{\bf 25},C^{23|01}_{\bf 36}$ \vspace{.1cm}\\\hline
$\mathtt{H}\left[\mathcal{V}_6\right]\setminus\left\{\mathbf{u}_1^3,\mathbf{u}_2^0,\mathbf{u}_3^3,\mathbf{u}_4^0,
\mathbf{u}_5^3,\mathbf{u}_6^0\right\}$ &~$C^{12|12}_{\bf 12},C^{13|02}_{\bf 23},C^{12|12}_{\bf 34},C^{13|02}_{\bf 45},C^{12|12}_{\bf 56},C^{13|02}_{\bf 61},C^{01|23}_{\bf 14},C^{23|01}_{\bf 25},C^{01|23}_{\bf 36}$ \vspace{.1cm}\\\hline
\end{tabular}
\caption{Four different Cabello sets $\mathtt{C^{[1]}_{18}}$, $\mathtt{C^{[2]}_{18}}$, $\mathtt{C^{[3]}_{18}}$, and $\mathtt{C^{[4]}_{18}}$ obtained from $\mathtt{P_{24}}=\mathtt{H}\left[\mathcal{V}_6\right]$. Left column lists the vectors and right column the contexts.}\label{tab2}
\end{table}
\noindent Notably, these sets are connected to each other through unitaries: 
\begin{subequations}
\begin{align}
&\mathtt{C^{[1]}_{18}}\xleftrightarrow{\mathrm U_1}\mathtt{C^{[2]}_{18}}~\&~\mathtt{C^{[3]}_{18}}\xleftrightarrow{\mathrm U_1}\mathtt{C^{[4]}_{18}};~~\mathtt{C^{[1]}_{18}}\xleftrightarrow{\mathrm U_2}\mathtt{C^{[4]}_{18}}~\&~\mathtt{C^{[2]}_{18}}\xleftrightarrow{\mathrm U_2}\mathtt{C^{[3]}_{18}};~~\mathtt{C^{[1]}_{18}}\xleftrightarrow{\mathrm U_3}\mathtt{C^{[3]}_{18}}~\&~\mathtt{C^{[2]}_{18}}\xleftrightarrow{\mathrm U_3}\mathtt{C^{[4]}_{18}};~~\text{where}\\
&\hspace{2cm}\scriptstyle \mathrm U_1:=
\begin{pmatrix}
0 &0 & 0 & -1\\
0 & 0 & -1 & 0\\
0 & 1 & 0 & 0\\
1 &0 & 0 & 0
\end{pmatrix};~~ \mathrm U_2:=
\begin{pmatrix}
0 & -1 & 0 & 0\\
1 & 0 & 0 & 0\\
0 & 0 & 0 & -1\\
0 & 0 & 1 & 0
\end{pmatrix};~~\mathrm U_3:=
\begin{pmatrix}
0 & 0 & -1 & 0\\
0 & 0 & 0 &1\\
1 & 0 & 0 & 0\\
0 & -1 & 0 &0
\end{pmatrix}.
\end{align}
\end{subequations}

\subsection{B-4. KS contextuality of the set $\mathrm{H}\left[\mathcal{V}^\prime_6\right]$ (Construction~\ref{const1})}

\begin{proposition}
The Hamilton extension $\mathrm{H}\!\left[\mathcal{V}^\prime_6\right]$ of the set
\begin{align*}
\mathcal{V}^\prime_6:=\left\{\begin{aligned}
&\vec u_1=(1,0,0,0);~~~~\vec u_2=(0,0,1,1);~~~~\vec u_3=(0,\sqrt{2},1,-1);\\
&\vec u_4=(0,\sqrt{2},-1,1);~\vec u_5=(0,\sqrt{2},1,1);~\vec u_6=(0,-\sqrt{2},1,1)
\end{aligned}\right\}
\end{align*}
is KS-uncolorable, thereby yielding a $24$-vector logical proof of the
Kochen--Specker theorem.
\end{proposition}

\begin{proof}
Since the vectors in $\mathcal{V}'_6$ are mutually Hamilton distinct, the Hamilton extension
$\mathrm{H}\!\left[\mathcal{V}'_6\right]$ contains the six Hamilton contexts

\begin{subequations}
\begin{align}
&C_{\bf 1}^{0123}\equiv\mathrm{H}[{\bf u}_1]=\Big\{{\bf u}^0_1=[(1,0,0,0)],~{\bf u}^1_1=[(0,1,0,0)],~{\bf u}^2_1=[(0,0,1,0)],~{\bf u}^3_1=[(0,0,0,1)]\Big\},\\
&C_{\bf 2}^{0123}\equiv\mathrm{H}[{\bf u}_2]=\Big\{{\bf u}^0_2=[(0,0,1,1)],~{\bf u}^1_2=[(0,0,1,-1)],~{\bf u}^2_2=[(1,-1,0,0)],~{\bf u}^3_2=[(1,1,0,0)]\Big\},\\
&C_{\bf 3}^{0123}\equiv\mathrm{H}[{\bf u}_3]=\Big\{{\bf u}^0_3=[(0,\sqrt{2},1,-1)],~{\bf u}^1_3=[(-\sqrt{2},0,1,1)],~{\bf u}^2_3=[(1,1,0,\sqrt{2})],~{\bf u}^3_3=[(1,-1,\sqrt{2},0)]\Big\},\\
&C_{\bf 4}^{0123}\equiv\mathrm{H}[{\bf u}_4]=\Big\{{\bf u}^0_4=[(0,\sqrt{2},-1,1)],~{\bf u}^1_4=[(\sqrt{2},0,1,1)],~{\bf u}^2_4=[(1,1,0,-\sqrt{2})],~{\bf u}^3_4=[(-1,1,\sqrt{2},0)]\Big\},\\
&C_{\bf 5}^{0123}\equiv\mathrm{H}[{\bf u}_5]=\Big\{{\bf u}^0_5=[(0,\sqrt{2},1,1)],~{\bf u}^1_5=[(\sqrt{2},0,1,-1)],~{\bf u}^2_5=[(1,-1,0,\sqrt{2})],~{\bf u}^3_5=[(1,1,-\sqrt{2},0)]\Big\},\\
&C_{\bf 6}^{0123}\equiv\mathrm{H}[{\bf u}_6]=\Big\{{\bf u}^0_6=[(0,-\sqrt{2},1,1)],~{\bf u}^1_6=[(\sqrt{2},0,-1,1)],~{\bf u}^2_6=[(-1,1,0,\sqrt{2})],~{\bf u}^3_6=[(1,1,\sqrt{2},0)]\Big\}.
\end{align}
\end{subequations}

\noindent Besides these six Hamilton contexts, the extension contains six Type (2--2) emergent contexts and twelve Type (1--1--1--1) emergent contexts. The latter alone suffice to derive a contradiction:
\begin{align}
\left\{\begin{aligned}
&C_{\bf 1234}^{0|0|0|0},~
C_{\bf 1234}^{1|1|1|1},~
C_{\bf 1234}^{2|2|2|2},~
C_{\bf 1234}^{3|3|3|3},~
C_{\bf 1256}^{0|1|0|0},~
C_{\bf 1256}^{1|0|1|1},\\
&C_{\bf 1256}^{2|3|2|2},~
C_{\bf 1256}^{3|2|3|3},~
C_{\bf 3456}^{0|1|3|2},~
C_{\bf 3456}^{1|0|2|3},~
C_{\bf 3456}^{2|3|1|0},~
C_{\bf 3456}^{3|2|0|1}
\end{aligned}\right\}.
\end{align}

\noindent Assume, for contradiction, that there exists a KS coloring
$\mu:\mathrm{H}\!\left[\mathcal{V}'_6\right]\rightarrow\{0,1\}$. The completeness rule requires exactly one projector in every context to be assigned value $1$, while the exclusivity rule forbids two orthogonal projectors from simultaneously receiving value $1$. The hypergraph $\mathrm{H}[\mathcal{V}'_6]$ possesses  symmetries which map equivalent branches of the coloring problem onto one another. It therefore suffices to analyze a single representative branch.

\noindent Choose $\mu({\bf u}_1^0)=1$. Then
\begin{subequations}\label{eq_AC_1}
\begin{align}
&\mu({\bf u}_1^1)=\mu({\bf u}_1^2)=\mu({\bf u}_1^3)=0,
&&\text{from }C_{\bf 1}^{0123},\\
&\mu({\bf u}_2^0)=\mu({\bf u}_3^0)=\mu({\bf u}_4^0)=0,
&&\text{from }C_{\bf 1234}^{0|0|0|0},\\
&\mu({\bf u}_2^1)=\mu({\bf u}_5^0)=\mu({\bf u}_6^0)=0,
&&\text{from }C_{\bf 1256}^{0|1|0|0}.
\end{align}
\end{subequations}

\noindent Since $C_{\bf2}^{0123}$ must contain exactly one projector assigned value $1$, the relations $\mu({\bf u}_2^0)=\mu({\bf u}_2^1)=0$ force exactly one of ${\bf u}_2^2$ and ${\bf u}_2^3$ to receive value $1$. By symmetry, we may choose $\mu({\bf u}_2^2)=1$, giving
\begin{subequations}\label{eq_AC_2}
\begin{align}
&\mu({\bf u}_2^3)=0,
&&\text{from }C_{\bf 2}^{0123},\\
&\mu({\bf u}_3^2)=\mu({\bf u}_4^2)=0,
&&\text{from }C_{\bf 1234}^{2|2|2|2},\\
&\mu({\bf u}_5^3)=\mu({\bf u}_6^3)=0,
&&\text{from }C_{\bf 1256}^{3|2|3|3}.
\end{align}
\end{subequations}

\noindent Likewise, the context $C_{\bf3}^{0123}$ together with $\mu({\bf u}_3^0)=\mu({\bf u}_3^2)=0$ forces exactly one of ${\bf u}_3^1$ and ${\bf u}_3^3$ to receive value $1$. By symmetry, choose $\mu({\bf u}_3^1)=1$. Consequently,
\begin{subequations}\label{eq_AC_3}
\begin{align}
&\mu({\bf u}_3^3)=0,
&&\text{from }C_{\bf 3}^{0123},\\
&\mu({\bf u}_4^1)=0,
&&\text{from }C_{\bf 1234}^{1|1|1|1},\\
&\mu({\bf u}_5^2)=0,
&&\text{from }C_{\bf 3456}^{1|0|2|3}.
\end{align}
\end{subequations}

\noindent Since $\mu({\bf u}_4^0)=\mu({\bf u}_4^1)=\mu({\bf u}_4^2)=0$, the completeness condition on $C_{\bf4}^{0123}$ forces $\mu({\bf u}_4^3)=1$, which in turn implies
\begin{align}\label{eq_AC_4}
\mu({\bf u}_5^1)=0,
\qquad\text{from }C_{\bf 3456}^{2|3|1|0}.
\end{align}

\noindent Finally,
\begin{align}
\mu({\bf u}_5^0)=
\mu({\bf u}_5^1)=
\mu({\bf u}_5^2)=
\mu({\bf u}_5^3)=0,
\end{align}
where the first, second, third, and fourth equalities follow from
Eqs.~\eqref{eq_AC_1}, \eqref{eq_AC_4}, \eqref{eq_AC_3}, and \eqref{eq_AC_2}, respectively. This contradicts the completeness condition for the Hamilton context $C_{\bf5}^{0123}$, which requires exactly one projector to be assigned value $1$. Therefore no KS coloring of $\mathrm{H}\!\left[\mathcal{V}'_6\right]$ exists. Since all remaining branches are related to the one above by symmetry, the contradiction is unavoidable. Hence $\mathrm{H}\!\left[\mathcal{V}'_6\right]$ is KS-uncolorable.
\end{proof}

\noindent {\bf Structural differences between $\mathrm{H}\left[\mathcal{V}'_6\right]$ and the Peres-24 construction:} Despite having the same numbers of vectors and measurement contexts, $\mathrm{H}\left[\mathcal{V}'_6\right]$ and $\mathtt{P_{24}}$ (i.e., $\mathrm{H}\left[\mathcal{V}_6\right]$ of Theorem~\ref{theo2}) possess markedly different structural properties. 
\begin{itemize}
\item[(1)] While both constructions contain 24 measurement contexts, their structural decompositions are different. Specifically, $\mathtt{P_{24}}$ consists of 6 Hamilton contexts and 18 Type (2-2) contexts, whereas $\mathrm{H}\left[\mathcal{V}'_6\right]$ consists of 6 Hamilton contexts, 6 Type (2-2) contexts, and 12 Type (1-1-1-1) contexts.

\item[(2)] $\mathtt{P}_{24}$ contains 18-vector subsets (the Cabello sets) that are themselves logically KS contextual. Moreover, the logical contradiction for this subset follows directly from a parity argument. In contrast, we find no subset of $\mathrm{H}\left[\mathcal{V}'_6\right]$ that exhibits logical KS contextuality through a parity-based proof. Nonetheless, $\mathrm{H}\left[\mathcal{V}'_6\right]$ contains 21-vector subsets that exhibit logical KS contextuality. For instance, the subset $\mathrm{H}\left[\mathcal{V}'_6\right]\setminus\left\{{\bf u}^0_1,{\bf u}^0_3,{\bf u}^0_5\right\}$ depicts logical KS contextuality.  
\end{itemize}

\section{Appendix C: Proof of Theorem~\ref{theo3}}  

\noindent Before presenting the proof, here we first note down an elementary but essential fact. As already pointed out, KS value assignment for a set of vectors $\mathcal{V}\subset\mathbb{C}^d$ has an graph theoretic formulation in terms of constrained vertex coloring of the associated orthogonality graph $\mathtt{G}_{\mathcal{V}}=(\mathtt{V},\mathtt{E})$. 

\begin{fact}\label{fact1}
Let $\mathtt{G}_{\mathcal{V}}=(\mathtt{V},\mathtt{E})$ be the orthogonality graph of a set of vectors $\mathcal{V}\subset\mathbb{C}^d$, and let $\mathtt{G}'=(\mathtt{V},\mathtt{E}')$ be a graph on the same vertex set with $\mathtt{E}\subseteq\mathtt{E}'$. Then every KS coloring of $\mathtt{G}'$ is also a KS coloring of $\mathtt{G}_{\mathcal{V}}$. Hence, if $\mathtt{G}'$ is KS colorable, so is $\mathtt{G}_{\mathcal{V}}$, and therefore $\mathcal{V}$ admits a $\{0,1\}$ KS value assignment.
\end{fact}

\noindent\underline{\bf Proof of Theorem~3}

\begin{proof}
Let $\mathcal{V}_{5}=\{{\bf A},{\bf B},{\bf C},{\bf D},{\bf E}\}\subset\mathbb{RP}^{3}$, 
with all projectors mutually Hamilton-distinct. Consequently, $\bigl|\mathtt{H}\left[\mathcal{V}_{5}\right]\bigr|=20$, where, for each ${\bf X}\in\{{\bf A},{\bf B},{\bf C},{\bf D},{\bf E}\}$, the Hamilton context is $C_{\bf X}^{0123}:=\mathtt{H}[{\bf X}]
=\left\{{\bf X}^{0},{\bf X}^{1},{\bf X}^{2},{\bf X}^{3}\right\}$. Consider a KS coloring $\mu:\mathtt{H}[\mathcal{V}_{5}]\rightarrow\{0,1\}$. For brevity, we use the following notations
\begin{align*}
\left.\begin{aligned}
C_{\bf X}^{01\ko{2}3}:=\left\{{\bf X}^{0},{\bf X}^{1},\ko{{\bf X}}^{\color{blue}{2}},{\bf X}^{3}\right\};~&\text{meaning that }\mu\left({\bf X}^{2}\right)=1\text{ and } \mu\left({\bf X}^{0}\right)=\mu\left({\bf X}^{1}\right)=\mu\left({\bf X}^{3}\right)=0\\[-2mm]
&\text{in the Hamilton context } C_{\bf X}^{0123},\\
C_{\bf UV}^{\ko{\alpha}\beta|\gamma\delta}
:=\left\{\ko{{\bf U}}^{\color{blue}{\alpha}},{\bf U}^{\beta},
{\bf V}^{\gamma},{\bf V}^{\delta}\right\};~&\text{meaning that }\mu\left({\bf U}^{\alpha}\right)=1 \text{ and all other projectors in type }\\[-2mm]
&\text{(2-2) emergent context }C_{\bf UV}^{\alpha\beta|\gamma\delta}\text{ are assigned }0,\text{ and}\\
C_{\bf PQRS}^{\alpha|\beta|\gamma|\ko{\delta}}
:=\left\{{\bf P}^\alpha,{\bf Q}^{\beta},{\bf R}^{\gamma},\ko{\bf S}^{\color{blue}{\delta}}\right\};~&\text{meaning that }\mu\left({\bf S}^{\delta}\right)=1 \text{ and all other projectors in type }\\[-2mm]
&\text{(1-1-1-1) emergent context }C_{\bf PQRS}^{\alpha|\beta|\gamma|\delta}\text{ are assigned }0.
\end{aligned}\right\}.
\end{align*}
We classify all possible emergent contexts that can occur in $\mathtt{H}[\mathcal{V}_5]$ and show that every such configuration admits a KS coloring. The proof proceeds by considering the following two cases: (I) the emergent contexts are exclusively of type (2-2); and (II) both type (1-1-1-1) and type (2-2) emergent contexts are present. 

\medskip
\noindent {\bf (I)} {\it Emergent contexts are exclusively of type (2-2):} Since the Hamilton extension is invariant under relabeling of the projectors, we may, without loss of generality, choose ${\bf A}$ as a reference projector. We therefore consider the type (2-2) emergent contexts involving ${\bf A}$ and ${\bf X}\in\{{\bf B},{\bf C},{\bf D},{\bf E}\}$ to be of the form $C_{\bf AX}^{\mu\nu|\alpha\beta}$. Furthermore, since the labels within each Hamilton context may be permuted arbitrarily, we may assume without loss of generality that $\alpha,\beta=0,1$. The possible choices of $(\mu,\nu)$ lead, up to relabeling, to four inequivalent configurations:
\begin{itemize}[leftmargin=2.5cm]
\item[{\bf Case-1.}] $\left\{C_{\bf AB}^{01|01},~C_{\bf AC}^{01|01},~C_{\bf AD}^{01|01},~C_{\bf AE}^{01|01}\right\}$;\qquad {\bf Case-2.} $\left\{C_{\bf AB}^{01|01},~C_{\bf AC}^{01|01},~C_{\bf AD}^{01|01},~C_{\bf AE}^{02|01}\right\}$;
\item[{\bf Case-3.}] $\left\{C_{\bf AB}^{01|01},~C_{\bf AC}^{01|01},~C_{\bf AD}^{02|01},~C_{\bf AE}^{02|01}\right\}$;\qquad {\bf Case-4.} $\left\{C_{\bf AB}^{01|01},~C_{\bf AC}^{01|01},~C_{\bf AD}^{02|01},~C_{\bf AE}^{03|01}\right\}$.
\end{itemize}

\noindent {\bf Case-1.} In this case, by Lemma~\ref{t3l5}, two type (2-2) emergent contexts arise between every pair of Hamilton contexts associated with ${\bf B}$, ${\bf C}$, ${\bf D}$, and ${\bf E}$. Observation~\ref{t3o1} uniquely fixes the form of these contexts. The $20$-projector configuration thus consists of $20$ type (2-2) emergent contexts and five Hamilton contexts, and admits a valid KS coloring
\begin{align*}
\left\{\begin{aligned}
&C_{\bf A}^{\ko{0}123},C_{\bf B}^{01\ko{2}3},C_{\bf C}^{01\ko{2}3},C_{\bf D}^{01\ko{2}3},C_{\bf E}^{01\ko{2}3},C_{\bf AB}^{\ko{0}1|01},C_{\bf AB}^{23|\ko{2}3},C_{\bf AC}^{\ko{0}1|01},C_{\bf AC}^{23|\ko{2}3},C_{\bf AD}^{\ko{0}1|01},C_{\bf AD}^{23|\ko{2}3},C_{\bf AE}^{\ko{0}1|01},C_{\bf AE}^{23|\ko{2}3},\\
&~~~C_{\bf BC}^{01|\ko{2}3},C_{\bf BC}^{\ko{2}3|01},C_{\bf BD}^{01|\ko{2}3},C_{\bf BD}^{\ko{2}3|01},C_{\bf BE}^{01|\ko{2}3},C_{\bf BE}^{\ko{2}3|01},C_{\bf CD}^{01|\ko{2}3},C_{\bf CD}^{\ko{2}3|01},C_{\bf CE}^{01|\ko{2}3},C_{\bf CE}^{\ko{2}3|01},C_{\bf DE}^{01|\ko{2}3},C_{\bf DE}^{\ko{2}3|01}
\end{aligned}\right\}.
\end{align*}

\noindent{\bf Case-2.} In this case, Observation~\ref{t3o1} uniquely determines the two type (2-2) emergent contexts between each of the pairs $({\bf B},{\bf C})$, $({\bf B},{\bf D})$, and $({\bf C},{\bf D})$. Then we can parameterize the above orthogonality relation as follows:

\begin{align*}
\left.\begin{aligned}
\mathrm{H}[{\bf A}]&\equiv\big\{[(1,0,0,0)],~[(0,1,0,0)],~[(0,0,1,0)],~[(0,0,0,1)]\big\},\\
\mathrm{H}[{\bf B}]&\equiv\big\{[(0,0,b_2,b_3)],~[(0,0,-b_3,b_2)],~[(-b_2,b_3,0,0)],~[(b_3,b_2,0,0)]\big\},\\
\mathrm{H}[{\bf C}]&\equiv\big\{[(0,0,c_2,c_3)],~[(0,0,c_3,-c_2)],~[(-c_2,c_3,0,0)],~[(c_3,c_2,0,0)]\big\},\\
\mathrm{H}[{\bf D}]&\equiv\big\{[(0,0,d_2,d_3)],~[(0,0,d_3,-d_2)],~[(-d_2,d_3,0,0)],~[(d_3,d_2,0,0)]\big\},\\
\mathrm{H}[{\bf E}]&\equiv\big\{[(0,e_1,0,e_3)],~[(e_1,0,e_3,0)],~[(0,e_3,0,-e_1)],~[(-e_3,0,e_1,0)]\big\}
\end{aligned}\right\}.
\end{align*}

\noindent This parameterization ensures that there cannot be a type (2-2) emergent context between $\{\bf B, C, D\}$ and $\bf E$.
Together with the five Hamilton contexts, they define a $20$-projector configuration that is KS colorable. An explicit KS coloring is provided below:
\begin{align*}
\left\{\begin{aligned}
&C^{\ko{0}123}_{\bf A},C^{01\ko{2}3}_{\bf B},C^{012\ko{3}}_{\bf C},C^{01\ko{2}3}_{\bf D},C^{012\ko{3}}_{\bf E},C_{\bf AB}^{\ko{0}1|01},C_{\bf AB}^{23|\ko{2}3},C_{\bf AC}^{\ko{0}1|01},C_{\bf AC}^{23|2\ko{3}},C_{\bf AD}^{\ko{0}1|01},\\
&~~C_{\bf AD}^{23|\ko{2}3},C_{\bf AE}^{\ko{0}2|01},C_{\bf AE}^{13|2\ko{3}},
C_{\bf BC}^{01|2\ko{3}},C_{\bf BC}^{\ko{2}3|01},C_{\bf BD}^{01|\ko{2}3},C_{\bf BD}^{\ko{2}3|01},C_{\bf CD}^{01|\ko{2}3},C_{\bf CD}^{2\ko{3}|01}
\end{aligned}\right\}.
\end{align*}

\noindent{\bf Case-3.} In this case, Observation~\ref{t3o1} uniquely determines the two type (2-2) emergent contexts corresponding to the pairs $({\bf B},{\bf C})$ and $({\bf D},{\bf E})$.Then we can parameterize the above orthogonality relation as follows:

\begin{align*}
\left.\begin{aligned}
\mathrm{H}[{\bf A}]&\equiv\big\{[(1,0,0,0)],~[(0,1,0,0)],~[(0,0,1,0)],~[(0,0,0,1)]\big\},\\
\mathrm{H}[{\bf B}]&\equiv\big\{[(0,0,b_2,b_3)],~[(0,0,-b_3,b_2)],~[(-b_2,b_3,0,0)],~[(b_3,b_2,0,0)]\big\},\\
\mathrm{H}[{\bf C}]&\equiv\big\{[(0,0,c_2,c_3)],~[(0,0,c_3,-c_2)],~[(-c_2,c_3,0,0)],~[(c_3,c_2,0,0)]\big\},\\
\mathrm{H}[{\bf D}]&\equiv\big\{[(0,d_1,0,d_3)],~[(d_1,0,d_3,0)],~[(0,d_3,0,-d_1)],~[(-d_3,0,d_1,0)]\big\},\\
\mathrm{H}[{\bf E}]&\equiv\big\{[(0,e_1,0,e_3)],~[(e_1,0,e_3,0)],~[(0,e_3,0,-e_1)],~[(-e_3,0,e_1,0)]\big\}
\end{aligned}\right\}.
\end{align*}

\noindent This parameterization ensures that there cannot be a type (2-2) emergent context between $\{\bf B, C\}$ and $\{\bf D,E\}$.
Together with the five Hamilton contexts, they define a $20$-projector configuration that is KS colorable. An explicit KS coloring is provided below:
\begin{align*}
\left\{\begin{aligned}
&C^{\ko{0}123}_{\bf A},C^{01\ko{2}3}_{\bf B},C^{01\ko{2}3}_{\bf C},C^{012\ko{3}}_{\bf D},C^{012\ko{3}}_{\bf E},C_{\bf AB}^{\ko{0}1|01},C_{\bf AB}^{23|\ko{2}3},C_{\bf AC}^{\ko{0}1|01},C_{\bf AC}^{23|\ko{2}3},\\
&~~C_{\bf AD}^{\ko{0}2|01},C_{\bf AD}^{23|2\ko{3}},C_{\bf AE}^{\ko{0}2|01},C_{\bf AE}^{13|2\ko{3}},C_{\bf BC}^{01|\ko{2}3},C_{\bf BC}^{\ko{2}3|01},C_{\bf DE}^{01|2\ko{3}},C_{\bf DE}^{2\ko{3}|01}
\end{aligned}\right\}.
\end{align*}

\noindent{\bf Case-4.}
In this case, $\mathbb{R}^4$ representation of $\mathrm{H}[{\bf U}]\equiv\left\{{\bf U}^0=[\vec u^0],{\bf U}^1=[\vec u^1],{\bf U}^2=[\vec u^2],{\bf U}^3=[\vec u^3]\right\}$ for  ${\bf U}\in\left\{\mathbf{A},\mathbf{B},\mathbf{C},\mathbf{D},\mathbf{E}\right\}$ enforces their forms to be 
\begin{align*}
\left.\begin{aligned}
\mathrm{H}[{\bf A}]&\equiv\big\{[(1,0,0,0)],~[(0,1,0,0)],~[(0,0,1,0)],~[(0,0,0,1)]\big\},\\
\mathrm{H}[{\bf B}]&\equiv\big\{[(0,0,b_2,b_3)],~[(0,0,-b_3,b_2)],~[(-b_2,b_3,0,0)],~[(b_3,b_2,0,0)]\big\},\\
\mathrm{H}[{\bf C}]&\equiv\big\{[(0,0,c_2,c_3)],~[(0,0,c_3,-c_2)],~[(-c_2,c_3,0,0)],~[(c_3,c_2,0,0)]\big\},\\
\mathrm{H}[{\bf D}]&\equiv\big\{[(0,d_1,0,d_3)],~[(d_1,0,d_3,0)],~[(0,d_3,0,-d_1)],~[(-d_3,0,d_1,0)]\big\},\\
\mathrm{H}[{\bf E}]&\equiv\big\{[(0,e_1,e_2,0)],~[(-e_1,0,0,e_2)],~[(e_2,0,0,e_1)],~[(0,e_2,-e_1,0)]\big\}
\end{aligned}\right\}.
\end{align*}
This parametrization satisfy all orthogonality relations and uniquely determined (2-2) context between $(\mathbf{B},\mathbf{C})$. However, no (2-2) context can be formed between any of the pairs $(\mathbf{B},\mathbf{D})$, $(\mathbf{B},\mathbf{E})$, $(\mathbf{C},\mathbf{D})$, $(\mathbf{C},\mathbf{E})$, or $(\mathbf{D},\mathbf{E})$. Together with the five Hamilton contexts, they define a $20$-projector configuration that is KS colorable. An explicit KS coloring is provided below:
\begin{align*}
\left\{\begin{aligned}
&C^{\ko{0}123}_{\bf A},C^{01\ko{2}3}_{\bf B},C^{01\ko{2}3}_{\bf C},C^{012\ko{3}}_{\bf D},C^{012\ko{3}}_{\bf E},C_{\bf AB}^{\ko{0}1|01},C_{\bf AB}^{23|\ko{2}3},C_{\bf AC}^{\ko{0}1|01},C_{\bf AC}^{23|\ko{2}3},\\
&~~\hspace{1cm}C_{\bf AD}^{\ko{0}2|01},C_{\bf AD}^{23|2\ko{3}},C_{\bf AE}^{\ko{0}3|01},C_{\bf AE}^{12|2\ko{3}},C_{\bf BC}^{01|\ko{2}3},C_{\bf BC}^{\ko{2}3|01}
\end{aligned}\right\}.
\end{align*}

\medskip
\noindent {\bf (II)} {\it Both type (1-1-1-1) and type (2-2) emergent contexts are present:} Without loss of generality, assume that the type (1-1-1-1) context $C^{0|0|0|0}_{\bf ABCD}$ is there. We distinguish the following two cases. Any other configuration is equivalent to one of these under relabeling and/or corresponds to a subgraph obtained by deleting orthogonality relations.

\medskip
\noindent{\bf Case-1:} The set of contexts is given by 
\begin{align*}
\left\{\begin{aligned}
&C^{\ko{0}123}_{\bf A},C^{01\ko{2}3}_{\bf B},C^{0\ko{1}23}_{\bf C},C^{012\ko{3}}_{\bf D},C^{01\ko{2}3}_{\bf E},C^{\ko{0}|0|0|0}_{\bf ABCD},C^{1|1|\ko{1}|1}_{\bf ABCD},C^{2|\ko{2}|2|2}_{\bf ABCD},C^{3|3|3|\ko{3}}_{\bf ABCD},\\
&\hspace{1cm}C^{\ko{0}1|01}_{\bf AB},C^{23|\ko{2}3}_{\bf AB},C^{0\ko{1}|01}_{\bf CD},C^{23|2\ko{3}}_{\bf CD},C^{\ko{0}1|01}_{\bf AE},C^{23|\ko{2}3}_{\bf AE},C^{01|\ko{2}3}_{\bf BE},C^{\ko{2}3|01}_{\bf BE}
\end{aligned}\right\}.
\end{align*}
Lemma~\ref{t3l8} guarantees that no further orthogonality relations can arise among the corresponding projectors. A valid KS-coloring of this vector set is shown. Consequently, any configuration obtained from it by relabeling vertices or by removing orthogonality relations is also KS-colorable (due to Fact.\ref{fact1}).

\medskip
\noindent{\bf Case 2:} In this case, assume that $\Vec y^\mu \perp \Vec w^\mu,~~
\mu\in\{0,1,2,3\}$. The resulting set of contexts is
\begin{align*}
\left\{\begin{aligned}
&C^{\ko{0}123}_{\bf A},C^{01\ko{2}3}_{\bf B},C^{0\ko{1}23}_{\bf C},C^{012\ko{3}}_{\bf D},C^{01\ko{2}3}_{\bf E},C^{\ko{0}|0|0|0}_{\bf ABCD},C^{1|1|\ko{1}|1}_{\bf ABCD},C^{2|\ko{2}|2|2}_{\bf ABCD},C^{3|3|3|\ko{3}}_{\bf ABCD},\\
&\hspace{3.5cm}C^{\ko{0}1|01}_{\bf AB},C^{23|\ko{2}3}_{\bf AB},C^{0\ko{1}|01}_{\bf CD},C^{23|2\ko{3}}_{\bf CD}
\end{aligned}\right\}.
\end{align*}
Lemma~\ref{t3l10} guarantees that no additional orthogonality relations can occur among the corresponding projectors. As a valid KS-coloring is already shown, therefore every graph obtained from this configuration by relabeling vertices and/or deleting orthogonality relations is likewise KS-colorable.

\medskip
\noindent Together, these cases exhaust all possible contextual configurations generated by five Hamilton-distinct projectors. As each resulting structure is KS-colorable, the claim follows.   
\end{proof}

\end{document}